\documentclass[a4paper,UKenglish,cleveref, autoref, thm-restate]{lipics-v2021}

\pdfoutput=1 
\hideLIPIcs  

\newif\iflong
\newif\ifshort
\newif\ifstar
\newif\ifnostar
\longtrue

\iflong
\else
\shorttrue
\fi

\ifstar
\else
\nostartrue
\fi

\title{Polynomial Kernel and Incompressibility for Prison-Free Edge Deletion and Completion} 

\titlerunning{Polynomial Kernel and Incompressibility for Prison-Free Edge Deletion and Completion} 

 \author{Séhane Bel Houari-Durand}{ENS Lyon, France}{sehane.bel_houari-durand@ens-lyon.fr}{0009-0003-6661-5171}{}
 
 \author{Eduard Eiben}{Department of Computer Science, Royal Holloway University of London, UK}{eduard.eiben@rhul.ac.uk}{0000-0003-2628-3435}{}

 \author{Magnus Wahlström}{Department of Computer Science, Royal Holloway University of London, UK}{magnus.wahlstrom@rhul.ac.uk}{0000-0002-0933-4504}{}

\authorrunning{S. Bel Houari-Durand, E. Eiben, and M. Wahlström} 

\Copyright{Séhane Bel Houari-Durand, Eduard Eiben, and Magnus Wahlström} 



\category{} 

\relatedversion{} 

\nolinenumbers 

\EventEditors{Olaf Beyersdorff, Micha\l{} Pilipczuk, Elaine Pimentel, and Nguyen Kim Thang}
\EventNoEds{4}
\EventLongTitle{42nd International Symposium on Theoretical Aspects of Computer Science (STACS 2025)}
\EventShortTitle{STACS 2025}
\EventAcronym{STACS}
\EventYear{2025}
\EventDate{March 4--7, 2025}
\EventLocation{Jena, Germany}
\EventLogo{}
\SeriesVolume{327}
\ArticleNo{21}

\usepackage{xspace}
\usepackage{todonotes}

\newtheorem{reductionRule}{Reduction Rule}

\newcommand{\cclass}[1]{\ensuremath{\mbox{\textup{#1}}}\xspace}
\newcommand{\NP}{\cclass{NP}}
\newcommand{\coNP}{\cclass{coNP}}

\newcommand{\containment}{\NP~$\subseteq$~\coNP/poly\xspace}

\newcommand{\delProblem}{\textsc{Prison-Free Edge Deletion}\xspace}

\newcommand{\cmd}{\ensuremath{\mathrm{cmd}}}

\newcommand{\FFF}{\ensuremath{\mathcal{F}}}
\newcommand{\PPP}{\ensuremath{\mathcal{P}}}
\newcommand{\SSS}{\ensuremath{\mathcal{S}}}
\newcommand{\cG}{\mathcal{G}\xspace}
\newcommand{\cH}{\mathcal{H}\xspace}
\newcommand{\cmsG}{\mathcal{F}}

\newcommand{\bigoh}{\mathcal{O}}

\newcommand{\typeOne}[1]{\ensuremath{\mathcal{T}^1_{#1}}}
\newcommand{\typeThree}[1]{\ensuremath{\mathcal{T}^2_{#1}}}

\begin{document}

\maketitle

\begin{abstract}
    Given a graph $G$ and an integer $k$, the \textsc{$H$-free Edge
      Deletion} problem asks whether there exists a set of at most $k$
    edges of $G$ whose deletion makes $G$ free of induced copies of
    $H$. Significant attention has been given to the \emph{kernelizability} 
    aspects of this problem -- i.e., for which graphs $H$ does the problem
    admit an ``efficient preprocessing'' procedure, known as a \emph{polynomial kernelization},
    where an instance $I$ of the problem with parameter $k$ is reduced
    to an equivalent instance $I'$ whose size and parameter value are bounded polynomially in $k$? 
    Although such routines are known for many graphs $H$ where the class
    of $H$-free graphs has significant restricted structure, 
    it is also clear that for most graphs $H$ the problem is \emph{incompressible},
    i.e., admits no polynomial kernelization parameterized by $k$ unless the polynomial hierarchy collapses.
    These results led Marx and Sandeep to the conjecture that \textsc{$H$-free
      Edge Deletion} is incompressible for any graph $H$ with at least five vertices, unless $H$
    is complete or has at most one edge (JCSS 2022). This conjecture was reduced
    to the incompressibility of \textsc{$H$-free Edge Deletion} for a
    finite list of graphs $H$. We consider one of these graphs, which we dub the
    \emph{prison}, and show that \textsc{Prison-Free Edge Deletion}
    has a polynomial kernel, refuting the conjecture. On the other
    hand, the same problem for the complement of the prison is
    incompressible.   
\end{abstract}
\newpage
\section{Introduction}

Let $H$ be a graph. A graph $G$ is \emph{$H$-free} if it does not
contain $H$ as an induced subgraph. More generally, let $\cH$ be a
collection of graphs. A graph is \emph{$\cH$-free} if it is
$H$-free for every $H \in \cH$. In the \textsc{$H$-free Edge Editing}
(respectively \textsc{$\cH$-free Edge Editing}) problem, 
given a graph $G$ and an integer $k$, the task is to add or
delete at most $k$ edges from $G$ such that the result is $H$-free
(respectively $\cH$-free). 
The \textsc{Edge Deletion} and \textsc{Edge Completion} variants
are the variants where only deletions, respectively only adding edges
is allowed. These are special cases of the much more general
\emph{graph modification} problem class, where a problem is defined by
a graph class $\cG$ and the natural problem variants
(\textsc{$\cG$ Edge Editing/Deletion/Completion} respectively
\textsc{$\cG$ Vertex Deletion}) are where the input graph $G$ is to be
modified so that the result is a member of $\cG$.

As Cai~\cite{Cai96} noted, for every finite $\cH$, the $\cH$-free
graph modification problems are FPT with a running time of $O^*(2^{O(k)})$
by a simple branching algorithm. However, the question of \emph{kernelization}
is much more subtle. A \emph{polynomial kernelization}
for a parameterized problem is a polynomial-time procedure
that takes as input an instance of the problem,
for example, $I=(G,k)$ in the case of a graph modification problem,
and outputs an instance $(G',k')$ of the same problem
such that $|V(G')|, k' \leq p(k)$ for some polynomially bounded
function $p(k)$ of $k$, and such that $(G',k')$ is a
yes-instance if and only if $(G,k)$ is a yes-instance. 
If so, we say that the problem has a \emph{polynomial kernel}. 
This has been used as a way to capture the notion of
efficient instance simplification and preprocessing,
and deep and extensive work has been done on determining
whether various parameterized problems have polynomial kernels
or not (under standard complexity-theoretical assumptions).
See the book by Fomin et al.~\cite{FominLSZkernelsbook}.

For many graph modification problems, both those characterized by
finite and infinite families $\cH$, polynomial kernelization is known,
but for many others, the question is wide open; see the survey of
Crespelle et al.~on parameterized edge modification
problems~\cite{CrespelleDFG23}.
For the structurally simpler case of \textsc{$H$-free Edge Deletion},
if $H$ is a clique then the problem reduces
to \textsc{$d$-Hitting Set} for $d=|E(H)|$
and has a polynomial kernel by the sunflower lemma,
and if $|E(H) \leq 1$ then the problem is trivial.
For the same reason, \textsc{$H$-free Vertex Deletion}
has a polynomial kernel for every fixed $H$. 
But in all other cases, the question is more intricate,
since deleting an edge in one copy of $H$ in $G$
can cause another copy of $H$ to occur, implying
a dependency between modifications that is not present
in the simpler cases. Beyond cliques and near-empty graphs,
polynomial kernels are known when $H$ is $P_3$ (i.e., $H$-free graphs are cluster graphs)~\cite{GrammGHN05},
$P_4$ (i.e., $H$-free graphs are cographs)~\cite{GuillemotHPP13}, the paw~\cite{2} and the diamond~\cite{CaoRSY22} (see Figure~\ref{fig:prison}).
Kernels are also known for several simple classes characterized by finite sets
$\cH$. But there are significant open cases; Crespelle et al.~\cite{CrespelleDFG23}
highlight the classes of claw-free graphs and line graphs, although the case of line graphs has since been resolved~\cite{EibenLochet20}.
Initially, progress led Fellows et al.~\cite{FellowsLRS07} to ask
(very speculatively) whether $\cH$-free graph modification problems
have polynomial kernels for all finite $\cH$.
This was refuted by Kratsch and Wahlström~\cite{KratschW13hfree},
and after a series of lower bounds, most importantly by Cai and Cai~\cite{3},
the answer now appears to be the opposite -- the
$\cH$-free graph modification problems have polynomial kernels
only for particularly restrictive choices of $\cH$.
Furthermore, in all such cases the kernel depends
intimately on the structural characterization
of the graph class, such as structural decomposition results.
However, it would appear unlikely that such a structural characterization 
of $H$-free graphs should exist for any arbitrary graph $H$, and correspondingly,
we would expect $H$-free edge modification problems 
not to admit polynomial kernelization.
For example, despite the above-mentioned positive results, \textsc{$H$-free Edge Deletion} has no kernel
for $H=P_\ell$ where $\ell \geq 5$, for $H=C_\ell$ for $\ell \geq 4$
or for any $H$ such that $H$ or its complement is 3-connected
(excepting the trivial cases)~\cite{3}.
Marx and Sandeep~\cite{1} pushed the pendulum in the other direction
and conjectured that for graphs $H$ on at least five vertices,
only the above-mentioned immediate kernels exist, conjecturing the following.

\begin{conjecture}[Conjecture~2 of~{\cite{1}}]
  \textsc{$H$-free Edge Deletion} does not have a polynomial kernel
  for any graph $H$ on at least 5 vertices, unless
  $H$ is complete or has at most one edge.
\end{conjecture}

They showed that this conjecture is equivalent to the statement
that \textsc{$H$-free Edge Deletion} admits no polynomial kernelization
if $H$ is one of nineteen specific graphs on five or six vertices. 
They also gave a corresponding conjecture for \textsc{$H$-free Edge Editing}
(where $|E(H)|=1$ is no longer a trivial case), and the case
of \textsc{$H$-free Edge Completion} follows by dualization.

In this paper, we refute this conjecture. We study the graph $H$ shown
in Figure~\ref{fig:prison} (the complement of $P_3+2K_1$), which we dub
a \emph{prison} (given that it can be drawn as the 5-vertex
``house'' graph with additional crossbars added). 
This is the first graph in the set $\cH$ in~\cite{1}. 
We show that \textsc{Prison-free Edge Deletion} has a polynomial
kernel. On the other hand, \textsc{Prison-free Edge Completion}
admits no polynomial kernelization unless the polynomial hierarchy
collapses. We leave \textsc{Prison-free Edge Editing} open for future work.

\subsection{Our results}

As expected, our result builds on a characterization of prison-free
graphs. We then derive the kernelization and lower bound results
working over this characterization. 

\subparagraph*{Prison-free graphs.} The start of our results is a
structure theorem for prison-free graphs. To begin with, note that
the 4-vertex induced subgraphs of a prison are $K_4$,
the diamond and the paw; see Figure~\ref{fig:prison}. The structure of diamond-free
and paw-free graphs are known: A graph is diamond-free if and
only if it is \emph{strictly clique irreducible}, i.e., every edge of
the graph lies in a unique maximal clique \cite{WallisZ90},
and a graph is paw-free if and only if every connected component
is either triangle-free or complete multipartite~\cite{4}. 
The structure of prison-free graphs generalizes both.

A \emph{complete multipartite graph} with classes $P_1, \ldots, P_m$
is a graph whose vertex set is the disjoint union $P_1 \cup \ldots \cup P_m$
and with an edge $uv$ for $u \in P_i$ and $v \in P_j$ if and only if $i \neq j$. 
Note that a clique is a complete multipartite graph where every part
is a singleton. We show the following. 
\begin{restatable}{theorem}{thmstructure} 
\label{the:handbag}
  A graph $G=(V,E)$ is prison-free if and only if the following holds:
  Let $F \subseteq V$ be an inclusion-wise maximal set such that
  $G[F]$ is complete multipartite with at least 4 parts, and
  let $v \in V \setminus F$.
  Then $N(v)$ intersects at most one part of $F$. 
\end{restatable}

Furthermore, let $cmd_p(G)$ for $p \in \mathbb{N}$ be the collection of all inclusion-wise
maximal $F \subseteq V(G)$ such that $G[F]$ is complete multipartite
with at least $p$ classes. We show that $cmd_4(G)$ induces a form of
partition of the cliques of $G$. 

\begin{restatable}{corollary}{corstructure}
\label{cor:cmd4}
  Let $G=(V,E)$ be a prison-free graph. The following hold.
  \begin{enumerate}
  \item If $F, F' \in cmd_4(G)$ are distinct and $F \cap F' \neq \emptyset$,
    then $F \cap F'$ intersects only one class $C$ of $F$ and
    $C'$ of $F'$, and there are no edges between $F \setminus C'$ 
    and $F' \setminus C$.
    In particular, $G[F \cap F']$ is edgeless.
  \item If $F, F' \in cmd_4(G)$ with $F \cap F' = \emptyset$,
    then for every $v \in F'$, $N(v)$ intersects at most one class of $F$
  \item Let $e \in E(G)$ be an edge that occurs in at least one $K_4$ in $G$. 
    Then there is a unique $F \in \cmd_4(G)$ such that $e$ occurs in $G[F]$. 
  \end{enumerate}
\end{restatable}

\medskip

In particular, every $K_p$ in $G$, $p \geq 4$ is contained in a unique $F \in \cmd_4(G)$.
It also follows that $cmd_4(G)$ can be enumerated in polynomial time in prison-free graphs. 

\subparagraph*{Lower bound for prison-free completion.}
We show that \textsc{Prison-free Edge Completion} is incompressible, i.e.,
admits no polynomial kernel parameterized by $k$ unless the polynomial hierarchy collapses.
Counterintuitively, this result
exploits the property that minimum solutions for the problem can be
\emph{extremely expensive} -- we can design a sparse graph $G$ such that
every prison-free supergraph of $G$ contains $\Theta(n^2)$ edges
(for example, by ensuring that the only possible $cmd_4$-decomposition
of a prison-free supergraph of $G$ consists of a single component $F=V(G)$). 
More strongly, we use this to show an additive \emph{gap hardness}
version of the problem.

\begin{restatable}{theorem}{thmgap} \label{thm:gap}
  For any $\varepsilon > 0$, it is NP-hard to approximate 
  \textsc{Prison-free Edge Completion}
  up to an additive gap of $g=\Theta(n^{2-\varepsilon})$,
  even if $G$ has an edge $e$ such that $G-e$ is $K_4$-free.
\end{restatable}

\medskip

With this in place, we can proceed with the lower bound using standard
methods, using the notion of cross-composition~\cite{BodlaenderJK14,FominLSZkernelsbook}.
We follow the method used in previous lower bounds against kernelization
of $H$-free edge modification problems~\cite{KratschW13hfree,3}.
Given a list $I_1, \ldots, I_t$ of instances of the above gap-version of
\textsc{Prison-Free Edge Completion} with parameter value $k$, our task is to produce an instance
of \textsc{Prison-Free Edge Completion} with parameter $(k'+\log t)^{O(1)}$
which corresponds to the OR of the instances $I_i$.
For this, we define a binary tree of
height $O(\log t)$ and place the instances at the leaves of the tree.
At the root of the tree, we place a single induced prison. 
For the internal nodes, we design \emph{propagational gadgets}
with the function that if the gadget at the node is edited,
then one of the gadgets at the children of the node must be edited as
well. Finally, for every instance $I_j=(G_j,k)$ with an edge $e_j$
such that $G_j-e$ is prison-free, we connect $e_j$ to the corresponding
gadget at the leaf of the tree and remove $e_j$ from $G_j$.
This forces at least one edge $e_j$, $j \in [t]$
to be added to the resulting graph
and the original instance $I_j=(G_j,k)$ must be solved.

The crux is that, unlike previous proofs (for example in Cai and Cai~\cite{3}
when $H$ is 3-connected) we cannot ``control'' the spread of the edge
completion solution to be confined to $G_j$. On the contrary, the
solution must spread all the way to the root and incorporate all
vertices from gadgets on the root-leaf path of the binary tree 
into a single complete multipartite component of the resulting
graph $G'$. Thus, we have no tight control over the number
of edges added in the corresponding solution. However, by the strong
lower bound on gap-hardness of Theorem~\ref{thm:gap}
we do not need tight control -- we can simply set the additive gap
$g$ large enough that the number of edges added outside of $G_j$ in
the resulting propagation is a lower-order term compared to $g$. 

\begin{restatable}{theorem}{thmprisonlb}
  \label{the:prisedit}
  \textsc{Prison-free Edge Completion} does not have a polynomial
  kernel parameterized by $k$ unless the polynomial hierarchy collapses. 
\end{restatable}

\subparagraph*{Kernel for prison-free deletion.}
The kernelization algorithm depends directly on the structural characterization of prison-free graphs. 
We start by using the sunflower lemma to obtain a small set $\PPP'$ of prisons in the graph $G$ such that any set of edges of size at most $k$ that intersects all prisons in $\PPP'$ has to intersect all prisons in $G$. We let $S$ be the set of vertices of these prisons. Note that $|S|= \bigoh(k^8)$, and outside of $S$ we only need to be concerned by prisons that are created by deleting some edge from $G$. This lets us delete all edges not in $E(G[S])$ that do not belong to a strict supergraph of a prison. In addition, if a prison in $G$ contains a single edge inside $S$, then such an edge has to be included in every solution of size at most $k$, so it can be deleted and $k$ decreased. After exhaustive application of these two reduction rules, we show that the edges of $G[V(G)\setminus S]$ can be partitioned into maximal complete multipartite subgraphs of $G[V(G)\setminus S]$ (even those which do not occur in $cmd_4(G[V(G)\setminus S])$). For each of these maximal complete multipartite subgraphs, we can check whether it is in some larger complete multipartite subgraph $F$ in $G$ that is nicely separated from the rest of $G$, in the sense that every $x\in V(G)\setminus F$ neighbors vertices in at most one class of $F$. For any such $F$, no supergraph of a prison can contain edge both outside and inside of $F$ and edges inside of $F$ can be safely deleted from $G$. This allows us to bound the number of maximal complete multipartite subgraphs outside of $S$ by $\bigoh(|S|^3)$. In addition, we show that any supergraph of a prison in $G$ that contains an edge fully outside of $S$ is fully contained in $S\cup F$ for a single maximal multipartite subgraph $F$ of $G[V(G)\setminus S]$. This allows us to treat these multipartite subgraphs separately. Moreover, using the fact that for any edge $e\in G[S]$, the graph $G[V(G)\setminus (S\setminus e)]$ is still prison-free, we can show that the interaction between $S$ and a maximal multipartite subgraph $F$ of $G[V(G)\setminus S]$ is very structured. This allows us to reduce the size of each of these subgraphs and we obtain the following theorem. 

\begin{restatable}{theorem}{thmkernel} 
    \delProblem admits a polynomial kernel.
\end{restatable}
\subparagraph*{Structure of the paper.}
In Section~\ref{sec:prison} we derive the basic facts about
prison-free graphs. Section~\ref{sec:lb} contains the lower bound
against \textsc{Prison-free Edge Completion}
and Section~\ref{sec:kernel} contains the polynomial kernel for
\textsc{Prison-free Edge Deletion}. We conclude in Section~\ref{sec:conc}.
\section{Structure of prison-free graphs}
\label{sec:prison}

We begin our study by characterizing the structure of prison-free
graphs. This generalizes two structures. The most closely related is
for \emph{paw-free graphs}; note that the paw is the subgraph of the
prison produced by deleting an apex vertex. Olariu~\cite{4} showed
that a graph is paw-free if and only if every connected component is
either triangle-free or a complete multipartite graph. 
The paw-free graph modification problems have polynomial kernels due
to Eiben et al.~\cite{2}. 
Second, less closely related but still relevant, the \emph{diamond}
$K_4-e$ is a subgraph of the prison  produced by deleting a vertex of
degree 3. It is known that a graph is \emph{diamond-free} if and only
if every edge occurs in only one maximal clique~\cite{WallisZ90}. 
The diamond-free graph modification problems have polynomial kernels
due to Cao et al.~\cite{CaoRSY22}.
In a sense, the structure of prison-free graphs generalizes both, as
the edge-sets of cliques $K_p$, $p \geq 4$ in a prison-free graph $G$
decomposes into complete multipartite induced subgraphs of $G$. 

We first prove Theorem~\ref{the:handbag} from the introduction,
then use it to derive the more informative Corollary~\ref{cor:cmd4}.

\begin{figure}
  \centering
\begin{subfigure}[1]{0.3\textwidth}
  \begin{tikzpicture}
    \foreach \x in {1,...,5}
    \node[circle, draw] (\x) at ({(- \x -2)*72 + 17}:1.2) {};
    
    \draw (5) -- (4) -- (1) -- (2) -- (3) -- (5);
    \draw (5) -- (1) -- (3) -- (4);
  \end{tikzpicture}
\end{subfigure}
\begin{subfigure}[1]{0.3\textwidth}
  \begin{tikzpicture}
    \foreach \x in {2,...,5}
    \node[circle, draw] (\x) at ({(- \x -2)*72 + 17}:1.2) {};
    
    \draw (5) -- (4);  \draw (2) -- (3) -- (5);
    \draw (3) -- (4);
  \end{tikzpicture}
\end{subfigure}
\begin{subfigure}[1]{0.3\textwidth}
  \begin{tikzpicture}
    \foreach \x in {1,...,4}
    \node[circle, draw] (\x) at ({(- \x -2)*72 + 17}:1.2) {};
    
    \draw (4) -- (1) -- (2) -- (3);
    \draw (1) -- (3) -- (4);
  \end{tikzpicture}
\end{subfigure}
\caption{Three graphs: The prison, the paw, and the diamond (as subgraphs of the prison).}
  \label{fig:prison}
\end{figure}
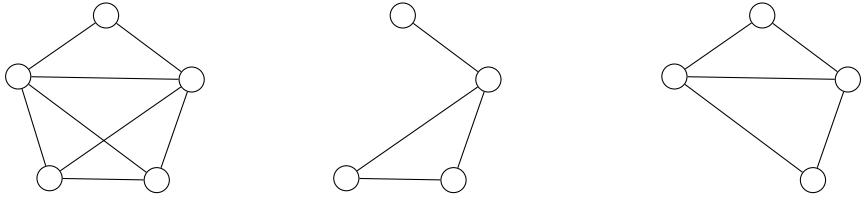

\iflong
\begin{lemma} \label{lem:kcomp}
\fi
\ifshort
\begin{lemma}[{$\star$\protect\footnote{Results marked with $\star$ have their proofs deferred to the full version.}}] \label{lem:kcomp}
\fi
Let $G$ be a prison-free graph and let $K \subseteq V(G)$ induce a $K_4$ in $G$. 
  Then there is a complete multipartite graph $G[F]$ in $G$ with $K \subseteq F$. 
  Furthermore, if $F$ is maximal under this condition, then for any $v \in V(G) \setminus F$, 
  $N(v)$ intersects at most one component of $F$.
\end{lemma}
\iflong
\begin{proof}
  We note that $F$ exists by assumption, since a clique is a complete
  multipartite graph. Let $F \subseteq V(G)$ be an arbitrary maximal
  set such that $K \subseteq F$ and $G[F]$ is complete multipartite.
  Write the components of $F$ as $F=P_1 \cup \ldots P_p$, and note $p \geq 4$.
  Now assume that there is a vertex $v\in V(G) \setminus F$ such that
  $N(v)$ intersects at least 2 components of $F$.
  Note that if there are components $P_1, \ldots, P_4$ (up to reordering)
  and vertices $u_1 \in P_1$, \ldots, $u_4 \in P_4$ such that
  $u_1, u_2 \in N(v)$ and $u_3, u_4 \notin N(v)$
  then $\{v,u_1,\ldots,u_4\}$ induces a prison. We claim that we can
  find such a configuration. For every component $P_i$, $i \in [p]$,
  either $P_i$ is disjoint from $N(v)$, or $P_i \subset N(v)$,
  or $P_i$ is \emph{mixed}. Order the components of $F$ to begin with 
  components contained in $N(V)$, then mixed components, then components  
  disjoint from $N(v)$. Then we can find a prison configuration unless
  either $P_2$ is disjoint from $N(v)$ or $P_{p-1}$ is included in $N(v)$.
  But the former does not happen by assumption. In the latter case,
  note that $P_p$ must be mixed as otherwise $F \cup \{v\}$ is
  complete multipartite. But then let $u_i \in P_i$ for $i=1,2$
  and select $w, w' \in P_p$ such that $w \in N(v)$ and $w' \notin N(v)$.
  Then $\{u_1,u_2,v,w,w'\}$ induces a prison. 
\end{proof}
\fi

Inspired by this, for any graph $G$ and $p \geq 2$,
define $\cmd_p(G)$ to consist of all maximal subsets $F \subseteq V$
such that $G[F]$ is complete multipartite with at least $p$ parts. 
We refer to $\cmd_p(G)$ as the \emph{complete multipartite decomposition}
of $G$ (and we will note that it is indeed a decomposition for $p \leq 4$
if $G$ is prison-free).
The following theorem is now an extension of the previous lemma.
Proofs of Theorem~\ref{the:handbag} and Corollary~\ref{cor:cmd4}
are deferred to the full version. 

\thmstructure*
\iflong
\begin{proof}
  One implication is by Lemma~\ref{lem:kcomp}. For the other direction,
  let $K=\{v_1,\ldots,v_4\}$ be such that $G[K]$ is a clique and
  $v \in V \setminus K$ such that $G[K \cup \{v\}]$ is a prison.
  Then there is a complete multipartite component $F_i$, $i \in [m]$
  such that $K \subseteq F_i$, and clearly the neighbours of $v$ in $K$
  are in distinct parts of $F_i$.
\end{proof}
\fi
We will use this to derive a more useful description of prison-free
graphs. 

\corstructure*
\iflong
\begin{proof}
  For the first, let $F, F' \in \cmd_4(G)$ be distinct sets. First
  assume that $F \cap F'$ covers at least three classes; we note that
  the partition of $F \cap F'$ into classes is identical in $F$ and $F'$.
  Let $v \in F \setminus F'$. Then there exist vertices $u_1, u_2 \in F \cap F'$
  such that $v$, $u_1$ and $u_2$ occur in distinct classes. Then $v \in V \setminus F'$
  sees two distinct classes of $F'$, a contradiction by Lemma~\ref{lem:kcomp}.
  Otherwise, if $F \cap F'$ contains two classes, let $u_1, u_2$ occur
  in distinct classes of $F \cap F'$ and let $v$ be a vertex of a class
  of $F$ not represented in $F \cap F'$. Then again $v \in V \setminus F'$
  sees two distinct classes of $F'$. Thus $F \cap F'$ can only represent
  a single class of $F$ and $F'$. This implies that $G[F \cap F']$ is edgeless.

  Furthermore, assume that there is an edge $uv \in E(G)$ such that
  $u \in F \setminus C'$ and $v \in F' \setminus C$ where
  $F \cap F' =  C \cap C'$ is non-empty for classes $C$ in $F$ and $C'$ in $F'$.
  Let $x \in F \cap F'$. Then $x, u \in N(v)$ which contradicts Lemma~\ref{lem:kcomp}.

  The second item is a direct consequence of Lemma~\ref{lem:kcomp}
  and the third is a consequence of the first, since any edge
  $e=uv \in E(G[F]) \cap E(G[F'])$ would imply that $u$ and $v$
  are in distinct classes of $F$ and $F'$ but $u, v \in F \cap F'$.    
  Note that at least one $F$ containing the clique exists by definition.
\end{proof}
\fi

In particular, the last item implies that every $K_p$, $p \geq 4$ is
contained in a unique multipartite component $F \in \cmd_4(G)$.
Since every $F \in \cmd_4(G)$ contains a $K_4$, and each
maximal component $F$ can be found greedily, $cmd_4(G)$ 
can be computed efficiently.

\section{Incompressibility of Prison-free Edge Completion}

\label{sec:lb}
In this section, we show that \textsc{Prison-free Edge Completion} admits
no polynomial kernel unless the polynomial hierarchy collapses.
The proof is in two parts. First we show a strong inapproximability result -- it is NP-hard 
to approximate \textsc{Prison-free Edge Completion} within an additive
gap of $g=O(n^{2-\varepsilon})$ for every $\varepsilon > 0$, even
for graphs with prison-free edge deletion number 1. 
We then use this to show a cross-composition (see below) for \textsc{Prison-free Edge Completion},
thereby ruling out polynomial kernels under standard complexity conjectures. 
This latter part roughly follows the outline of previous proofs of incompressibility~\cite{KratschW13hfree,3}.

\subsection{Initial observations and support gadgets}
\label{sec:subprel}
\newcommand{\mincext}{\ensuremath{\mathrm{mincext}}}

We begin with some useful statements.

\begin{proposition}
  For a complete multipartite graph $K$ with parts of sizes $a_1,a_2,...,a_m$, the number of edges of $K$ is 
  $\frac{1}{2}\sum_{\substack{i\neq j}}a_ia_j = \frac{1}{2}(|K|^2-\sum_{\substack{i}}a_i^2)$.
\end{proposition}


For a graph $G=(V,E)$ and a set of edges $A$ over $V$, we let $G \cup A$ denote the graph $G'=(V, E \cup A)$. 
A \emph{prison-free completion set} for $G$ is an edge set $A$ over $V(G)$ such that $G \cup A$ is prison-free.
A \emph{solution} to $(G,k)$ is a prison-free completion set $A$ for $G$ with $|A| \leq k$. 
The following is essential in our lower bounds. 

\iflong
\begin{lemma} \label{lem:cardone}
\fi
\ifshort
\begin{lemma}[$\star$] \label{lem:cardone}
\fi
  Let $G$ be a graph with exactly one induced $K_4$ and let $A$ be a minimal prison-free completion set for $G$. 
  Then $\cmd_4(G \cup A)$ has exactly 1 component and $A$ lives within that component.
\end{lemma}\iflong
  \begin{proof}
    Let $F$ be the element of $\cmd_4(G \cup A)$ that contains the four vertices of the induced $K_4$ of $G$, which is unique by Corollary~\ref{cor:cmd4}.
    Let $A'=A\cap F^2$. Then every $K_4$ in $G \cup A'$ is contained in $F$, and $G \cup A'$ is prison-free.

    Indeed, for the first, let $H=(G \cup A')[S]$ induce a $K_4$ and let $u \in S \setminus F$.
    Since $G[S]$ is not a $K_4$ by assumption, $H$ contains an edge from $A'$,
    hence there are at least two distinct vertices $v, w \in S \cap F$.
    But $F$ is complete multipartite and $u \notin F$ sees vertices
    from at least two different parts of $F$. The same will clearly hold
    in $G \cup A$, which (by Theorem~\ref{the:handbag}) contradicts 
    that $G \cup A$ is prison-free. Hence $H$ does not exist. By the same argument, any prison
    in $G \cup A'$ must be contained in $F$. But $(G \cup A')[F]=(G \cup A)[F]$
    is prison-free by assumption.
    Since $A$ is minimal, we have $A=A'$ and the claim follows.
  \end{proof}
\fi





We next show a way to enforce \emph{forbidden edges}, i.e., non-edges $uv$ in $G$ such that no prison-free completion set for $G$ of at most $k$ edges contains $uv$. 


\iflong
\begin{lemma}
\fi
\ifshort
\begin{lemma}[$\star$]
\fi
  Let $G$ be a graph, $k\in \mathbb{N}$, and $u,v\in V(G)$ with $u \neq v$ such that $uv\notin E(G)$.
  There is a graph $G'$ on vertex set $V(G')=V(G) \cup F$ such that $G=G'-F$ and the following holds:
  the minimal solutions to $(G',k)$
  are precisely the minimal solutions $A$ to $(G,k)$ such that $uv \notin A$.
  Furthermore, every solution $A$ to $(G',k)$ satisfies $uv \notin A$. 
\end{lemma}
%
\iflong
  \begin{proof}
    Let $F=\{u_{ij} \mid i \in [k+1], j \in [3]\}$. Add the vertices $F$ to $G$ along with edges such that for each $i \in [k+1]$
    $\{u,u_{i1}, u_{i2}, u_{i3}\}$ induces a $K_4$ and $u_{i3}v$ is an edge. 
    The resulting graph $G'$ is the graph we need. 

    First let $A$ be a solution to $(G,k)$ with $uv \notin A$. We claim that $G' \cup A$ is also prison-free.
    Indeed, let $P \subseteq V(G')$ be such that $(G' \cup A)[P]$ is a prison. Then $P$ must contain vertices both of $F$ and $V(G)$.
    But then $\{u,v\}$ is an independent separator of $P$, which does not exist in a prison. 
    Hence $A$ is a solution to $(G',k)$. Conversely, let $A'$ be a minimal solution to $(G',k)$
    with $uv \notin A'$. Then the restriction $A=A' \cap V(G)^2$ is clearly a solution to $(G,k)$,
    and by the previous statement also a solution to $(G',k)$. Since $A'$ is minmial, $A'=A$ and $A'$ is vertex-disjoint from $F$.
    It follows that the minimal solutions $A$ to $(G,k)$ and $(G',k)$ with $uv \notin A$ coincide.

    On the other hand, let $A$ be a prison-free completion set for $G'$ with $uv \in A$. Then for all $i \in [k+1]$, $A$ must contain 
    $u_{i1}v$ or $u_{i2}v$ to prevent $u,u_{i1},u_{i2},u_{i3},v$ from being a prison, which implies $|A| > k$. 
\end{proof}
\fi


\begin{figure}
     \centering
     \captionsetup[subfigure]{justification=centering}
     \begin{subfigure}[t]{0.2\textwidth}
       \begin{tikzpicture}
         \foreach \x in {1,...,5}
         \node[circle, draw, inner sep=1pt] (\x) at ({(- \x -2)*72 + 17}:1.2) {\x};

         \draw (1) -- (3) -- (5) -- (2) -- (4) -- (1);
         \draw (5) -- (1) ;
         \draw (3) -- (4);

         \draw[dashed, dash pattern=on 4pt off 4pt, line width=1.5pt] (4) -- node[above] {$e_1$} (5);
         \draw[dashed, dash pattern=on 4pt off 4pt, line width=1.5pt] (1) -- node[above] {$e_2$} (2);
         \draw[dashed, dash pattern=on 4pt off 4pt, line width=1.5pt] (2) -- node[above] {$e_3$} (3);
         
       \end{tikzpicture}
       \caption{Propagational component}
       \label{fig:prop}
     \end{subfigure}
     \begin{subfigure}[t]{0.35\textwidth}
       \begin{tikzpicture}[scale=0.6, every node/.style={inner sep=1pt}]
  
   \node[circle, draw] (A) at (4,4) {$a_0$};
   \node[circle, draw] (B) at (6,4) {$b_3$};
   \node[circle, draw] (C) at (5,6) {$c_0$};
   
   \node[circle, draw] (D) at (2,4) {$a_1$};
   \node[circle, draw] (E) at (2,6) {$b_0$};
   \node[circle, draw] (F) at (0,5) {$c_1$};
   
   \node[circle, draw] (G) at (2,2) {$a_2$};
   \node[circle, draw] (H) at (0,2) {$b_1$};
   \node[circle, draw] (I) at (1,0) {$c_2$};
   
   \node[circle, draw] (J) at (4,2) {$a_3$};
   \node[circle, draw] (K) at (4,0) {$b_2$};
   \node[circle, draw] (L) at (6,1) {$c_3$};

  \draw (A) -- (B);
  \draw (A) -- (J);
  \draw (A) -- (L);

  \draw (B) -- (J);
  \draw (B) -- (L);

  \draw (C) -- (J);
  \draw (C) -- (L);

  \draw (D) -- (E);
  \draw (D) -- (A);
  \draw (D) -- (C);

  \draw (E) -- (A);
  \draw (E) -- (C);

  \draw (F) -- (A);
  \draw (F) -- (C);

  \draw (G) -- (H);
  \draw (G) -- (D);
  \draw (G) -- (F);

  \draw (H) -- (D);
  \draw (H) -- (F);

  \draw (I) -- (D);
  \draw (I) -- (F);

  \draw (J) -- (K);
  \draw (J) -- (G);
  \draw (J) -- (I);

  \draw (K) -- (G);
  \draw (K) -- (I);

  \draw (L) -- (G);
  \draw (L) -- (I);

   \draw[dashed, dash pattern=on 4pt off 4pt, line width=1.5pt] (A) -- node[above] {$e_1$} (C);
   \draw[dashed, dash pattern=on 4pt off 4pt, line width=1.5pt] (D) -- node[below] {$e_2$} (F);
   \draw[dashed, dash pattern=on 4pt off 4pt, line width=1.5pt] (G) -- node[below] {$e_3$} (I);
   \draw[dashed, dash pattern=on 4pt off 4pt, line width=1.5pt] (J) -- node[above] {$e_4$} (L);

   \draw[dotted, line width=1.5pt] (B) -- (C);
   \draw[dotted, line width=1.5pt] (E) -- (F);
   \draw[dotted, line width=1.5pt] (H) -- (I);
   \draw[dotted, line width=1.5pt] (K) -- (L);
 
 \end{tikzpicture}
 \caption{Cloning component of length 4}
  \label{fig:clon}
\end{subfigure}
\begin{subfigure}[t]{0.4\textwidth}
  \begin{tikzpicture}[scale=0.6, every node/.style={inner sep=1pt}]
  
   \node[circle, draw] (A) at (3,0) {1};
   \node[circle, draw] (B) at (5,0) {2};
   \node[circle, draw] (C) at (3,2) {3};
   \node[circle, draw] (D) at (5,2) {4};
   \node[circle, draw] (E) at (4,4) {5};
   
   \node[circle, draw] (F) at (6,5) {6};
   \node[circle, draw] (G) at (7,3) {7};
   \node[circle, draw] (H) at (8,5) {8};
   
   \node[circle, draw] (I) at (1,3) {9};
   \node[circle, draw] (J) at (2,5) {10};
   \node[circle, draw] (K) at (0,5) {11};

  \draw (A) -- (C);
  \draw (A) -- (D);
  \draw (A) -- (E);

  \draw (B) -- (C);
  \draw (B) -- (D);
  \draw (B) -- (E);

  \draw (C) -- (I);
  \draw (C) -- (J);
  \draw (C) -- (K);
  
  \draw (D) -- (F);
  \draw (D) -- (G);
  \draw (D) -- (H);

  \draw (E) -- (I);
  \draw (E) -- (J);
  \draw (E) -- (K);
  
  \draw (E) -- (F);
  \draw (E) -- (G);
  \draw (E) -- (H);

  \draw (C) -- (D);
  \draw (I) -- (J);
  \draw (F) -- (G);

   \draw[dashed, dash pattern=on 4pt off 4pt,  line width=1.5pt] (A) -- node[above] {$e_1$} (B);
   \draw[dashed, dash pattern=on 4pt off 4pt,  line width=1.5pt] (K) -- node[above] {$e_2$} (J);
   \draw[dashed, dash pattern=on 4pt off 4pt,  line width=1.5pt] (F) -- node[above] {$e_3$} (H);

   \draw[dotted, line width=1.5pt] (K) -- (I);
   \draw[dotted, line width=1.5pt] (G) -- (H);
   
 \end{tikzpicture}
 \caption{Disjoint propagational component}
  \label{fig:dis}
\end{subfigure}
\caption{Gadgets for Section~\ref{sec:lb}. Dotted lines are forbidden edges; dashed lines are named ``gadget-edges'' with special semantics.}
\end{figure}
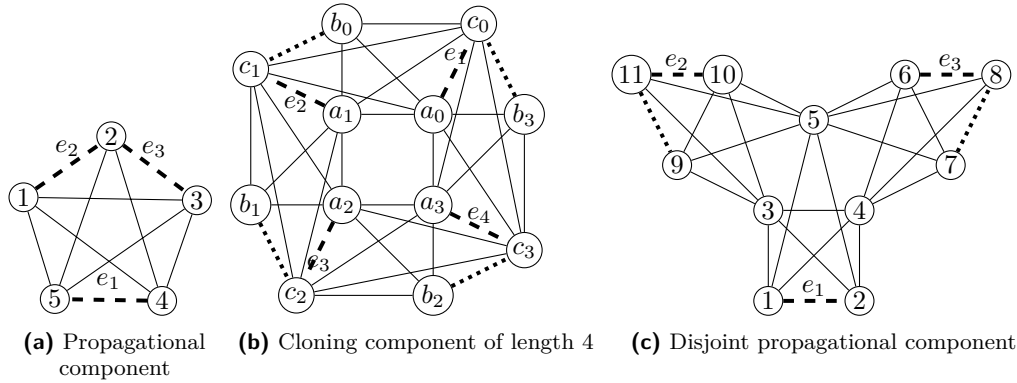

We now proceed with gadget constructions. 
A \emph{propagational component} is a graph containing 3 distinct non-edges $e_1$, $e_2$ and $e_3$ such that for any graph $G$ and subset $A$ of $V(G)^2$, if $G \cup A$ is prison-free and $e_1\in A$, then $e_2\in A$ or $e_3\in A$. 
Figure~\ref{fig:prop} shows this component. 
In what follows, we will deduce gadgets from this propagation property, leaving the proof of their prison-freeness for later. 

\begin{definition}
    Let $l\geq 4$. A \emph{cloning component of length $l$} is a component over the vertices $a_0,...,a_{l-1}$, $b_0,...,b_{l-1}$, $c_0,...,c_{l-1}$ with the edges such that for all $0\leq i \leq l-1$, $a_{i+1},c_{i+1},b_i,c_i,a_i$ induces a propagational component with $e_1^i,e_2^i,e_3^i=a_ic_i,$ $a_{i+1}c_{i+1},$ $c_{i+1}b_i$, and the edge $e_3^i=c_{i+1}b_i$ is forbidden. All arithmetic here is modulo $l$.
\end{definition}

A cloning component is drawn in Figure~\ref{fig:clon}. Note that for all $0\leq i \leq l-1$, if a solution $A$ for an instance $(G,k)$ contains $e_1^i$ in a cloning component, then it contains $e_2^i$ (since $e_2^i \in A$ or $e_3^i \in A$, but $e_3^i$ is forbidden). We inductively obtain the next property.

\iflong
\begin{lemma} \label{lem:cloning-comp-props}
\fi
\ifshort
\begin{lemma}[$\star$] \label{lem:cloning-comp-props}
\fi
  Let $l\geq 4$, $k\geq 1$ and $G$ be a graph containing an induced cloning component $X$ of size $l$ with vertices named $a_i$, $b_i$ and $c_i$ as above. Let $A$ be a subset of non-edges such that $|A|\leq k$ and $G \cup A$ is prison-free. Then either $\{a_ic_i \mid 0 \leq i \leq l-1\} \cap A=\emptyset$
  or $a_ic_i \in A$ for every $0 \leq i \leq l-1$.
  Furthermore, in the latter case all of $X$ is contained in a complete multipartite component of $G \cup A$.
\end{lemma}\iflong
\begin{proof}
  The former follows easily by induction (as noted above). The latter follows from Lemma~\ref{lem:cardone} applied to $X$ with a single edge $a_0c_0$ added. That is, $X \cup \{a_0c_0\}$ contains a unique $K_4$, thus $(G \cup A)[X]$
  contains a unique complete multipartite component $F$ which contains all vertices of $V(A) \cap X$. 
  This includes every vertex $a_i$ and $c_i$ of $X$, and furthermore every vertex $b_i$ of $X$
  has neighbours $a_i$ and $c_i$, which are in distinct components of $F$ since $a_ic_i \in A$. 
\end{proof}
In what follows, for a cloning component $X$, we will use the names $X.e_1$, \ldots, $X.e_l$ to refer to the edges as in Figure~\ref{fig:clon},
i.e., $X.e_i=a_{i-1}c_{i-1}$ for $i \in [l]$.
  \fi

We define one final gadget, shown in Figure~\ref{fig:dis}. This does the same job as a propagational component -- i.e., if $e_1 \in A$ then $e_2 \in A$ or $e_3 \in A$ -- except that the edges $e_1$, $e_2$, $e_3$ are pairwise vertex-disjoint. 
This is the \emph{propagational gadget} mentioned in the proof overview.

\begin{definition}
  A \emph{disjoint propagational component} is a graph isomorphic to the graph shown in Figure~\ref{fig:dis}, i.e.,
  a graph on a vertex set $V=\{v_1,\ldots,v_{11}\}$ such that vertex sets $\{v_1,\ldots,v_5\}$,
  $\{v_4,\ldots,v_8\}$ and $\{v_3,v_5,v_9,v_{10},v_{11}\}$ all induce propagational components,
  $v_1v_2$, $v_3v_5$, $v_4v_5$, $v_6v_8$ and $v_{10}v_{11}$ are standard non-edges and $v_7v_8$
  and $v_9v_{11}$ are forbidden edges.
  The edge labelled $e_1=v_1v_2$ is referred to as \emph{input edge} and the edges $e_2=v_{10}v_{11}$ and $e_3=v_6v_8$ are \emph{output edges}.
\end{definition}

%
%
%
%
%
%
%
%
%
%
%
%
%
%
%

\subsection{NP-hardness of Gap Prison-free Edge Completion}

We now prove the first half of the incompressibility result for \textsc{Prison-Free Edge Completion}, namely that it remains NP-hard even in a strong additive gap version.

Let \textsc{Gap Prison-free Edge Completion} be the variant of \textsc{Prison-free Edge Completion}
where the input is a triple $(G,k,g)$ and the task is to distinguish between the following cases:
\begin{enumerate}
\item $G$ has a prison-free completion set of at most $k$ edges
\item $G$ has no prison-free completion set of fewer than $k+g$ edges
\end{enumerate}
For intermediate cases (where the size $t$ of a minimum-cardinality prison-free editing set is $k < t < k+g$)
the output may be arbitrary. The following is the more precise version of Theorem~\ref{thm:gap} from the introduction.
\ifshort
  We defer the construction to the full version of the paper.
\fi


\ifshort
\begin{theorem}[\ifshort{$\star$}\fi] \label{thm:prison-free-gap}
\fi
\iflong
\begin{theorem} \label{thm:prison-free-gap}
\fi
  For any $\varepsilon > 0$, it is NP-hard to distinguish between yes-instances and no-instances $(G,k,g)$
  of \textsc{Gap Prison-free Edge Completion} even if $G$ contains 
  an edge $e$ such that $G-e$ is $K_4$-free
  and the gap is $g=\Theta(n^{2-\varepsilon})$.
\end{theorem}

\ifshort
\begin{proof}[Proof sketch]
    We can only give a brief sketch here and refer to the full version for details.
    The result is shown by a reduction from \textsc{Vertex Cover} on cubic graphs. 
    At its heart is the following principle. We create a graph $G$ which has a single induced $K_4$.
    Then by Lemma~\ref{lem:cardone}, for any minimal prison-free completion set $A$ of $G$, 
    $V(A)$ will be contained in a single complete multipartite component $F$ of $G \cup A$. 
    Using forbidden edges, we can have tight control over the partition of $F$,
    which lets us predict $|A|$ well. In particular, we can ensure that the number of edges $|A|$ added
    scales quadratically with $|V(A)|$. 
    
    Our reduction uses two gadgets: cloning components, of sufficiently large length $\ell$, 
    and disjoint propagational components. Let $(G,t)$ where $G=(V,E)$ be an input instance of \textsc{Vertex Cover}
    where $G$ is a cubic graph. Using one initial cloning component $X_0$ with a single seeded active edge $e$, we can 
    force the propagation of $e$ into $m=|E|$ disjoint activation edges $e_i$, one for every edge of $E$. 
    For every edge $ab \in E$, associated with an edge $e_i$ of $X_0$, we create a disjoint propagational component
    with input edge $e_i$ and output edges $f_{a,i}$ and $f_{b,i}$ associated with the vertices $a$ and $b$ of $G$.
    Finally, for every vertex $v \in V$ we create a very large cloning component, which contains all edges
    $f_{v,j}$ associated with it, and whose length $\ell$ depends on the desired gap $g$.  
    Thus, prison-free completion sets $A$ of the resulting graph $G'$ which activate the cloning components
    of $s$ distinct vertices of $G$, lead to a prison-free supergraph $G' \cup A$ of $G'$
    where the unique complete multipartite component $F$ of $cmd_4(G' \cup A)$ contains
    $O(\ell s)$ vertices, guaranteeing that $|A|=\Theta((s\ell)^2)$ for a corresponding minimum solution $A$. 
    We can now achieve the desired gap by tuning $\ell$ to be sufficiently large.
\end{proof}
\fi

\iflong
We show a reduction from \textsc{Vertex Cover} on cubic graphs, which is well known to be NP-complete~\cite{GareyJS74}. Let $(H,k)$ be an instance of \textsc{Vertex Cover} where $H$ is a cubic graph. 
Let $\ell \geq 6$ be a parameter which will control the gap value $g$; $\ell$ is even. 
Enumerate the vertices and edges of $H$ as $V_H=\{v_1,\ldots,v_n\}$ and $E_H=\{e_1,\ldots,e_m\}$.
We create an instance $(G,k,g)$ of \textsc{Gap Prison-free Edge Completion} as follows. 
\begin{enumerate}
\item Create a cloning component $X_0$ of length $m+1$ where every $X_0.e_i$ is a non-edge except for $X_0.e_{m+1}$ which is an edge.
\item For every vertex $v_i \in V_H$, create a cloning component $X_i$ of length $\ell$ and
select three pairs from $\{X_i.e_j\}$ at distance at least two from each other; refer to these pairs as $p_{i,j}q_{i,j}$, $j=1, 2, 3$.
\item For every edge $e_i=v_av_b$ in $E_H$, create a disjoint propagational component with input edge $X_0.e_j$
and output edges $p_{a,j_a}q_{a,j_a}$ and $p_{b,j_b}q_{b,j_b}$ where $j_a, j_b \in [3]$ are chosen so that
all occurrences of each vertex $v_i$ in $E_H$ correspond to distinct pairs $p_{i,j'}q_{i,j'}$.
\item Finally, forbid all edges within the following sets:
  \begin{itemize}
  \item $S_1=\{X_i.a_{2j} \mid i \in [n], 1 \leq j \leq \ell/2\}$
  \item $S_2=\{X_i.a_{2j-1} \mid i \in [n], 1 \leq j \leq \ell/2\}$
  \item $S_3=\{X_i.c_{2j} \mid i \in [n], 1 \leq j \leq \ell/2\} \cup \{X_i.b_{2j-1} \mid i \in [n], 1 \leq j \leq \ell/2\}$
  \item $S_4=\{X_i.b_{2j} \mid i \in [n], 1 \leq j \leq \ell/2\} \cup \{X_i.c_{2j-1} \mid i \in [n], 1 \leq j \leq \ell/2\}$
  \end{itemize}
\end{enumerate}
This completes the description of the graph $G$. 

To discuss the output parameters $k$ and $g$, we need to describe the structure of prison-free supergraphs $G \cup A$ of $G$. Note that $G$ has a unique induced $K_4$. Indeed, all gadget components added are $K_4$-free, except for one $K_4$ involving the edge $X_0.e_{m+1}$. Thus by Lemma~\ref{lem:cardone} any minimal prison-free supergraph $G \cup A$ of $G$ is such that $\cmd_4(G \cup A)$ contains a unique complete multipartite component $F$, and furthermore $V(A) \subseteq F$. 
Thus, the number of edges in $(G \cup A)[F]$ depends (as in Section~\ref{sec:subprel}) on $|F|$ and the sizes of the parts in the complete multipartite decomposition of $F$.

\begin{lemma} \label{lm:vc-lower}
  Let $G=(V,E)$ be the graph constructed above and let $A \subseteq V^2$ be a minimal edge set such that $G \cup A$ is prison-free
  and there are at least $t$ indices $i$ such that $p_{i,j_i}q_{i,j_i} \in A$ for some $j_i \in \{1,2,3\}$.
  Then $|A| \geq (13/4)(t \ell)^2-9t\ell$. 
\end{lemma}
\begin{proof}
  Let $I \subseteq [n]$ be the set of indices $i$ such that at least one pair $p_{i,j_i}q_{i,j}$ is contained in $A$.
  Then for every cloning component $X_i$, $i \in I$ and every $j \in [\ell]$, we have $X_i.e_j \in A$ and by Lemma~\ref{lem:cloning-comp-props}  
  all of $X_i$ is contained in some complete multipartite component of $\cmd_4(G \cup A)$.   
  As noted above, since $G$ has a unique $K_4$ and $A$ is minimal, Lemma~\ref{lem:cardone} implies that
  $\cmd_4(G \cup A)$ has a single complete multipartite component $F$. Thus $V(X_i) \subseteq F$ for every $i \in I$. 
  Let $V_I=\bigcup_{i \in I} V(X_i)$. Note that the cloning components $X_i$ are pairwise vertex-disjoint,
  thus $|V_I|=3t\ell$. 
  
  Furthermore, because of the forbidden edges, each vertex set $S_i$, $i=1, 2, 3, 4$ is represented in only one part of $F$,
  and since all six edges between different parts $S_i$ exist in every cloning component, the only possible partition of $V_I$
  induced by $F$ is as $(V_I \cap S_1, \ldots, V_I \cap S_4)$. 
  These parts have (respectively) cardinality $t\ell/2$, $t\ell/2$, $t\ell$ and $t\ell$. 
  Then, as observed in Section~\ref{sec:subprel} the number of edges of $(G \cup A)[V_i]$ is
  \[
    \frac{(3t\ell)^2 -2(t\ell/2)^2 - 2(t\ell)^2}{2} = \frac{13}{4} (t\ell)^2.
  \]
  The number of edges of $V_I$ already present in $G$ is at most $9t\ell$ since $|V_I|=3t\ell$ and every vertex in a cloning component has degree at most 6. 
  In addition, the cloning component $X_0$ must be present in $F$, as well as some parts of the disjoint propagational components; we refrain from analysing these in detail for a loose lower bound.
\end{proof}

In the other direction, we observe that the bound of Lemma~\ref{lm:vc-lower} is tight up to lower-order terms in $\ell$.

\begin{lemma} \label{lm:vc-upper}
  Assume that $H$ has a minimal vertex cover of cardinality $t$. Then $G$ has a prison-free completion set $A$
  with $|A| \leq (13/4)(t\ell)^2 + 36t \ell (m+1) + 144(m+1)^2$ where $m=|E_H|$. 
\end{lemma}
\begin{proof}
  Let $S \subseteq V_H$ be a minimal vertex cover of cardinality $t$. Let $V_I=\bigcup_{v_i \in S} V(X_i)$, $V_0=V(X_0)$
  and let $V_E$ contain all vertices of $G$ from disjoint propagational components; i.e., $V_0 \cup V_I \cup V_E$ contains
  every vertex of $G$ except for those internal to cloning components $X_i$ where $v_i \notin S$. 
  Note that $|V_0| =3(m+1)$ and $|V_E| \leq 11m$. Let $F=V_0 \cup V_I \cup V_E$ and initialize the construction of a partition of $F$
  as $((V_0 \cup V_I) \cap S_1, \ldots, (V_0 \cup V_I) \cap S_4)$. Note that every part induces an independent set.
  We claim that we can complete this into a partition by extending it to the vertices $V_E \setminus (V_0 \cup V_I)$.
  Indeed, consider a disjoint propagational component with input edge $e_1$ and output edges $e_2$ and $e_3$,
  and refer to Figure~\ref{fig:dis} for vertex numbers. Then $e_1$ is also an edge of $X_0$, hence $e_1 \in A$.
  Furthermore, the component represents some edge $v_av_b \in E_H$ and since $S$ is a vertex cover, at least one
  of $v_a$ and $v_b$ is in $S$. Also note that if $v_a \in S$, respectively $v_b \in S$, then $e_2 \subseteq V_I$
  respectively $e_3 \subseteq V_I$, and thus these vertices are already accounted for in the partition.
  We proceed as follows. Again, we take no care to optimize the constant factors except for the leading term $(t\ell)^2$. 
  \begin{itemize}
  \item If $v_a, v_b \in S$ then place $9$ with $11$, $7$ with $8$, and the remaining unplaced vertices $3, 4, 5$ in new parts
  \item If only one of $v_a$ and $v_B$ is in $S$, say the vertex represented by $e_2$, then $1, 2, 10$ and $11$ have been placed. Place $9$ with $11$. Place $3$ in a new part and $4,5$ together in another part. Leave $6, 7, 8$ outside of $F$.
  \item The remaining case that $v_a, v_b \notin S$ cannot happen since $S$ is a vertex cover. 
  \end{itemize}
  We verify that every part of this partition induces an independent set and that every vertex outside of $F$ has neighbours in at most one part of $F$.
  For the former observe that every edge of $G$ is contained in a cloning component or a disjoint propagational component. It is easy to verify that the partition $(S_i)_i$ of cloning components consists of four independent sets, and for disjoint propagational components note that $10$ and $11$ (if placed) respectively $6$ and $8$ (if placed) belong to different parts, hence the placement of vertices $7$ and $9$ does not create problems; and all other vertex placements are clearly safe.
  Next, for vertices outside of $F$, as noted, edges of $G$ are entirely contained in gadget components. Cloning components are pairwise vertex-disjoint, and every cloning component is either contained in $F$ or disjoint from $F$. Similarly, disjoint propagational components are pairwise vertex-disjoint, and each such component is either contained in $F$ or care was taken to ensure that vertices $6, 7, 8$ respectively $9, 10, 11$ see only one part of $F$. 

  Thus, if we add edges to $G$ to complete $F$ into a complete multipartite component, then the resulting graph $G \cup A$ is prison-free by Theorem~\ref{the:handbag}. As computed in Lemma~\ref{lm:vc-lower}, the number of edges in $(G \cup A)[V_I]$ is $(13/4)(t\ell)^2$. 
  In addition, the total number of edges possible with an endpoint in $F \setminus V_I$ is at most
  \[
    |V_I| \cdot |F \setminus V_I| + |F \setminus V_I|^2 \leq (3t\ell)(12(m+1)) + (12(m+1))^2,
  \]
  given that $|V(X_0)|=3(m+1)$ and that every distributional component contributes at most 9 further vertices.
  Clearly $|A|$ is at most this value. 
\end{proof}

\begin{proof}[Proof of Theorem~\ref{thm:gap}]
  Let $\ell$ be some even number and apply the construction above. Clearly it runs in polynomial time assuming $\ell$ itself is polynomially bounded.
  If $H$ has a vertex cover of size at most $t$, then by Lemma~\ref{lm:vc-upper} there is a prison-free completion set $A$
  with $|A| \leq (13/4)(t\ell)^2 + O(t \ell m+m^2)$. On the other hand, if every vertex cover of $H$ consists of at least $t+1$ vertices,
  then by Lemma~\ref{lm:vc-lower} every prison-free completion set $A$ satisfies
  \[
    |A| \geq   (13/4)((t+1)\ell)^2-O(t \ell) = (13/4)(t\ell)^2 + (13/2)t \ell^2 - O(t \ell).
  \]
  The gap between these values is of size $(13/2)t\ell^2-O(t \ell m + m^2)$. 
  Clearly, we can set $\ell$ large enough that the $\ell^2$-term in the first part dominates the second part. 
  Furthermore, if we want a gap of some $g=\Theta(n^{2-\varepsilon})$ then,
  since $\ell^2$ dominates the gap for large values of $\ell$, it suffices to use $\ell=\Theta(n^{1-\varepsilon/2})$.
  Since $|V(G)|=3|V_H|\ell + O(m)$, setting $\ell=\Theta(|V_H|^{2/\varepsilon})$ suffices, asymptotically.
\end{proof}

\fi 

\subsection{Compositionality of Prison-Free Edge Completion}

We prove Theorem \ref{the:prisedit} by an or-composition over instances of \textsc{Gap Prison-free Edge Completion},
using Theorem~\ref{thm:gap} to support the composition. 

We recall some definitions~\cite{BodlaenderJK14}. 
A \emph{polynomial equivalence relation} is an equivalence relation on
$\Sigma^*$ such that the following hold:
\begin{enumerate}
\item There is an algorithm that given two strings $x, y \in \Sigma^*$
  decides in time polynomial in $|x|+|y|$ whether $x$ and $y$ are
  equivalent.
\item For any finite set $S \subset \Sigma^*$, the number of
  equivalence classes that $S$ is partitioned to is polynomially
  bounded in the size of the largest element of $S$.
\end{enumerate}
Let $L \subseteq \Sigma^*$ be a language, $\mathcal{R}$ a polynomial
equivalence relation and $Q \subseteq \Sigma^* \times \mathbb{N}$ a
parameterized language. An \emph{OR-cross-composition of $L$ into $Q$
  (with respect to $\mathcal{R}$)}
is an algorithm that given $t$ instances $x_1, \ldots, x_t \in \Sigma^*$
of $L$ belonging to the same equivalence class of $\mathcal{R}$,
uses time polynomial in $\sum_{i=1}^t |x_i|$ and outputs
an instance $(y,k)$ of $Q$ such that the following hold:
\begin{enumerate}
\item The parameter value $k$ is polynomially bounded in $\max_i
  |x_i|+\log t$.
\item $(y,k)$ is a yes-instance of $Q$ if and only if at least one
  instance $x_i$ is a yes-instance of $L$.
\end{enumerate}
If an NP-hard language $L$ has an OR-cross-composition into a
parameterized problem $Q$ then $Q$ admits no polynomial kernelization,
unless the polynomial hierarchy collapses~\cite{BodlaenderJK14}.
We proceed to show this for \textsc{Prison-free Edge Completion}.

\ifshort
  We present here a high-level sketch of the result, based directly on Theorem~\ref{thm:gap}.
  In the full version, we present a more careful proof with explicit parameters.
  Thus for simplicity, let $(G_i, k_i, g_i)_{i=1}^t$ be a sequence of instances of \textsc{Gap Prison-free Edge Completion}
  with a sufficiently high gap value $g_i=\Theta(|V(G_i)|^{2-\varepsilon})$, $\varepsilon>0$, to be tuned later.
  Via a polynomial equivalence relation, we may assume that $|V(G_i)|=n_0$, $|E(G_i)|=m_0$, $k_i=k_0$ and $g_i=g$ holds for every input instance $i$. 
  By Theorem~\ref{thm:gap}, we assume that for every instance $(G_i, k_i, g_i)$ there is a single edge $e_i \in E(G_i)$
  such that $G_i-e_i$ is prison-free; refer to this as the \emph{activation edge} of $G_i$ an d delete it from the graph. 
\fi

\iflong
For transparency, we view Theorem~\ref{thm:gap} as a procedure that transforms an instance of \textsc{Vertex Cover} with cubic input graph to an instance of \textsc{Gap Prison-free Edge Completion} with a given gap parameter $g$.
Thus, let $(G_i, k_i)_{i=1}^t$ be a sequence of cubic instances of \textsc{Vertex Cover}.
Via a polynomial equivalence relation, we may assume that $|V(G_i)|=n_0$, $|E(G_i)|=m_0$ and $k_i=k_0$ holds for every input instance $i$. 
For each $i \in [t]$ and $\ell \in \mathbb{N}$, let $(G_{i,\ell},k',g)$ be the result of applying Theorem~\ref{thm:gap}
to the instance $(G_i,k_i)$ with parameter value $\ell$. Thus by Lemma~\ref{lm:vc-upper} $k'=(13/4)(k_0\ell)^2 + 36k\ell(m_0+1)+144(m_0+1)^2$ and we need to set $\ell$ large enough that $(13/4)((k_0+1)\ell)^2 > k'$. 
We refrain from fixing a precise value here, since the gap $g$ will need to overcome some additional slack due to the construction of an or-composition. Furthermore, for every $i \in [t]$, delete the unique edge of the cloning component $X_0$ in $G_{i,\ell}$ that is present from the start. Refer to this as the \emph{activation edge} of $G_{i,\ell}$. 
\fi

For the composition, let $h \in \mathbb{N}$ be such that $2^{h-1} < t \leq 2^h$. Define a balanced binary tree of height $h$ whose leaves are labelled $L_i$, $i=1, \ldots, t$. Please a disjoint propagational component for every internal node $x$ of the tree, identifying the two output edges with the children of $x$ and the input edge at $x$ with the corresponding output edge of the parent of $x$. Initially, all input and output edges are absent except that the input edge at the root of the tree is present.
\iflong
Finally identify the output edge leading into the leaf $L_i$ with the activation edge of the graph $G_{i,\ell}$. Let $e_i$ refer to this activation edge. 
\fi
\ifshort
Finally identify the output edge leading into the leaf $L_i$ with the activation edge $e_i$ of the graph $G_i$. 
\fi
If there are any unused output edges, for example if $t$ is odd, make these edges forbidden. 
Let $(G,k)$ be the resulting instance of \textsc{Prison-free Edge Completion} where $k$ is yet to be determined. 

\iflong
\begin{lemma} \label{lm:final-tree-recursion}
\fi
\ifshort
\begin{lemma}[$\star$] \label{lm:final-tree-recursion}
\fi
  Let $A$ be a prison-free completion set for $G$. Then there is some $i \in [t]$ such that $A$ contains the activation edge $e_i$ of the graph \iflong $G_{i,\ell}$. \fi \ifshort $G_i$. \fi Furthermore, for every $i \in [t]$, any minimal prison-free edge completion set $A_i$ for \iflong $G_{i,\ell}+e_i$ \fi \ifshort $G_i+e_i$ \fi
  can be completed into a prison-free completion set $A$ for $G$ which contains every output edge in the path from the root to $L_i$ but no other output edges from the tree. 
\end{lemma}\iflong
\begin{proof}
  The first claim holds easily by induction, starting from the input edge of the disjoint propagational component at the root. For the second, let $A_i$ be a prison-free completion set for $G_{i,\ell}$ containing the activation edge $e_i$ of this instance, but which is otherwise minimal.
  By construction, $G_{i,\ell}$ has a unique $K_4$, thus $\cmd_4(G_{i,\ell} \cup A_i)$ contains a single multipartite component $F$ with $V(A_i) \subseteq F$. We complete $F$ into a suitable partition that includes all input/output edges in the tree for ancestors of $L_i$ but no other input/output edges.  
  Let $x$ be an ancestor of $L_i$ and let $e$ be the output edge of its disjoint propagational component leading to $L_i$.
  We recursively make sure that the two vertices of $e$ are in distinct components of $F$. This guarantee holds initially since $A_i$ contains the activation edge of $G_{i,\ell}$. Thus, using names from Figure~\ref{fig:dis} and assuming w.l.o.g.\ that $e$ represents $e_2$ in the figure, we place vertex $9$ with vertex $11$, vertex $3$ in a new set, vertices $4$ and $5$ together in a new set, and vertices $1$ and $2$ in distinct new sets. Thus $F$ has been extended to include the input edge leading into $x$ as promised. Let $A$ consist of all edges that must be added to $G$ to complete $F$ into a complete multipartite component. As in Lemma~\ref{lm:vc-upper}, it is easy to verify (using that $A_i$ is a prison-free completion set for $G_{i,\ell}$) that every component of $F$ induces an independent set, that $A$ does not contain any forbidden edge, and that every vertex of $G$ not present in $F$ sees vertices of at most one part from $F$. 
\end{proof}
\fi
We are now ready to finish the proof.




\thmprisonlb*

\iflong
\begin{proof}
  Let $n_a=3k_0\ell$, $n_b=12(m_0+1)$, $n_c=6h$ and $g=(n_b+n_c)n_a + (n_b+n_c)^2$.
  Set $\ell$ high enough that $(13/4)((k_0+1)\ell)^2  - 9k_0 \ell > (13/4)(k_0\ell)^2+g$.
  Finally, let $G$ be the graph constructed above constructed with parameter $\ell$
  and set $k=(13/4)(k_0\ell)^2+g$. Our output is $(G,k)$.
  We first note that our parameter is not too large. Indeed, as in Theorem~\ref{thm:gap}
  the guaranteed gap grows quadratically with $\ell$ while the parameter $g$ above is linear in $\ell$
  and polynomial in $n_0+\log t$. It is also clear that the construction can be executed in polynomial time.
  Thus it only remains to show that $(G,k)$ is a yes-instance if and only if at least one input instance $(G_i,k_0)$ is a yes-instance. 

  For this, on the one hand, let $A$ be a minimal prison-free completion set for the output $G$, $|A| \leq k$. 
  By Lemma~\ref{lm:final-tree-recursion} there is at least one leaf $L_i$ such that $A$ contains the activation edge for instance $G_{i,\ell}$. Thus $A$ also contains a prison-free completion set for the ``activated'' instance $G_{i,\ell}$. Furthermore, by the choice of $g$ and $\ell$, by Lemma~\ref{lm:vc-lower} $A$ contains edges of at most $k$ cloning components $X_j$ in $G_{i,\ell}$. Thus the cloning components $X_j$ of $G_{i,\ell}$ which are active in $A$ represent a vertex cover of $G_i$ of cardinality at most $k_0$, as in the proof of Theorem~\ref{thm:gap}.

  On the other hand, assume that input instance $(G_i,k_0)$ is a yes-instance. Then by Theorem~\ref{thm:gap}
  there is a prison-free completion set $A_i$ for $G_{i,\ell}$ with $|A_i| \leq (13/4)(k_0\ell)^2 + n_an_b+n_b^2$.
  Furthermore, we may assume that $\cmd_4(G_i \cup A')$ contains a single complete multipartite component $F$ and that
  $|F| \leq n_a+n_b$. By Lemma~\ref{lm:final-tree-recursion} we can complete this into a complete multipartite component $F$ that also covers eight vertices of every disjoint propagational component in every ancestor of $L_i$ and no further vertices.
  This adds at most $n_c=6h$ vertices to $F$ (since the output- and input-edges of the components use the same vertices).
  Thus, the number of edges $A$ needed to complete $F$ into a complete multipartite component is at most $|A_i|+(n_a+n_b)n_c+n_c^2 \leq k$. Finally, $G \cup A$ is prison-free since (per Lemma~\ref{lm:final-tree-recursion}) no vertex outside of $F$ sees more than one part of $F$ and $G$ contains no $K_4$ not included in $F$. 

  Thus we have showed a cross-composition from \textsc{Vertex Cover} to \textsc{Prison-free Edge Completion} with parameter $k$, implying that the latter has no polynomial kernel unless \containment, in which case the polynomial hierarchy collapses.
\end{proof}
\fi    
\ifshort
  \begin{proof}[Proof sketch]
    We show that the construction above is a cross-composition into \textsc{Prison-free Edge Completion} with parameter $k$.
    Consider briefly the two cases. First, if for some $i \in [t]$ the input instance $(G_i, k_0, g_0)$ is positive,
    let $A_i$ be a prison-free completion set for $G_i$ with $|A_i| \leq k$. By Lemma~\ref{lm:final-tree-recursion}
    there is a prison-free completion set $A \supseteq A_i$ for $G$ that modifies edges along the root-leaf path
    to $L_i$ of the binary tree, and does not contain any other output edge of the tree. By Lemma~\ref{lem:cardone}
    and minimality of $A$ and $A_i$ we may now assume that $A$ touches no
    vertices of the binary tree except  on the root-leaf path to $L_i$,
    and contains no further edges inside $G_i$ beyond $A_i$. 
    Let $n_t$ be the number of vertices in the gadgets of the tree incident with edges of $A$; since the tree has height $h=O(\log t)$
    and each node is constant size, we have $n_t=O(\log t)$. Thus
    \[
      |A| \leq |A_i| + (n_h + n_0)n_h \leq k_0 + O((n_0 + \log t)\log t).
    \]
    On the other hand, assume that for every $i \in [t]$ the minimum prison-free completion set
    $A_i$ for $G_i$ has $|A_i| \geq k_0+g_0$. Let $A$ be a prison-free completion set for $G$.
    By Lemma~\ref{lm:final-tree-recursion} there is some $i \in [t]$ such that the activation
    edge $e_i$ of $G_i$ is contained in $A$. Thus $A$ contains a prison-free completion set for $G_i$
    and $|A| \geq k_0 + g_0$. By the construction of Theorem~\ref{thm:gap},
    we can tune the parameters so that these quantities separate, and choose a parameter $k$ 
    where
    \[
      k_0 + O((n_0 + \log t) \log t) \leq k < k_0 + g_0
    \]
    in which case the reduction is complete. 
  \end{proof}
\fi

\section{Polynomial kernel for Prison-free Edge Deletion}
\label{sec:kernel}
In this section we will find a polynomial kernel for \delProblem. 
Let $G$ be a graph and $k \geq 1$. We will fix $(G,k)$ throughout this section to be an instance of \delProblem.

Throughout this section, for a graph $G$, and an edge $e=uv$ of $G$, we call common neighborhood of $e$ the set $N_G(e)=N_G(u)\cap N_G(v)$.
In addition, given a graph $G$ and a set of vertices $S\subseteq V(G)$, we denote by $\bar{S}$ the set $V(G)\setminus S$, when $G$ is clear from the context.

\subsection{Finding a Small Vertex Modulator} 
We start by finding a small subset of vertices $S$ such that any edgeset $A\subseteq E(G)$ with at most $k$ edges that intersects all prisons in $G[S]$ also intersects all prisons in $G$. While this is not sufficient for a kernel, as deleting an edge can create a prison, outside of this set, we only need to focus on prisons that are created by deleting an edge. 
To obtain this set, we will use well known Sunflower Lemma due to Erd{\"o}s and Rado~\cite{Erdos60}.

A \emph{sunflower} in a set family $\FFF$ is a subset $\FFF' \subseteq \FFF$ such that all pairs of elements in $\FFF'$ have the same intersection called \emph{core}.

\begin{lemma}[Sunflower Lemma,\cite{Erdos60,FlumGrohe06}]\label{lem:SF}
  Let $\FFF$ be a family of subsets of a universe $U$, each of cardinality exactly
  $b$, and let $a \in \mathbb{N}$. If $|\FFF|\geq b!(a-1)^{b}$, then $\FFF$
  contains a sunflower $\FFF'$ of cardinality at least $a$. Moreover,
  $\FFF'$ can be computed in time polynomial in $|\FFF|$.
\end{lemma}

\iflong 
\begin{lemma}
\fi
\ifshort  
\begin{lemma}[$\star$]
\fi
\label{lem:kernelizing_P_F}
We can in in polynomial time either determine that $(G,k)$ is no-instance of \delProblem\ or compute a set $S\subseteq V(G)$ with $|S|\le 5\cdot 8!\cdot (k+1)^8$ such that for every prison $P$ in $G$ and every $A\subseteq E(G)$ with $|A|\le k$, it holds that if $G[S]\Delta A$ is prison-free, then $A\cap E(G[P])\neq \emptyset$. 

\end{lemma}
\ifshort
\begin{proof}[Proof Sketch]
    The proof is a straightforward application of Lemma~\ref{lem:SF}.
    Let $\SSS = \{E(P)\mid P\text{ is a prison in }G\}$. Note that each set $X\in \SSS_0$ contains precisely $8$ edges. We iteratively apply Lemma~\ref{lem:SF} on $\SSS$ to find a sunflower of size $k+2$. If the core of the sunflower $\SSS'$ is empty, then $G$ contains $k+2$ edge-disjoint prisons and $(G,k)$ is no instance. Else any solution of size at most $k$ interesects the core and that is the case even if we require to hit only $k+1$ prisons of $\SSS'$, so we remove arbitrary prison from $\SSS$ and repeat the procedure until no suflower of size $k+2$ can be find. At that point $\SSS$ contains at most $ 8!\cdot (k+1)^8$ many prisons and any $A\subseteq E(G)$ with $|A|\le k$ that intersect all of the prisons that are left in $\SSS$ intersects all prisons in $G$. We let $S$ to be the set of all vertices in these prisons. 
\end{proof}
\fi
\iflong 
\begin{proof}
    The proof is a straightforward application of Lemma~\ref{lem:SF}.
    Let $\SSS_0 = \{E(P)\mid P\text{ is a prison in }G\}$ and let us now define a sequence of length $q\ge |\SSS_0|\le |V(G)|^5$ ($q$ is to be determined later) of families of subsets of $E(G)$ interactively as follows:
    Given $\SSS_i$, $i\in \mathbb{N}$, if $|\SSS_i| < 8!\cdot (k+1)^8$, then we let $q=i$ and we stop the sequence. Else, notice that each set $S\in \SSS_i$ is a set of edges in a single prison and so $|S| = 8$. It follows from Lemma~\ref{lem:SF} that $\SSS_i$ contains a sunflower $\SSS'_i$ with $|\SSS'_i| \ge k+2$ that can be computed in polynomial time. If the prisons in $\SSS'_i$ are pairwise edge-disjoint, then we stop and return that $(G,k)$ is no-instance. This is because in this case we need at least $k+2$ edges to hit all prisons of $\SSS'_i$. Else,
    let $S_i$ be an arbitrary set of edges representing a prison in $\SSS_i'$. We set $\SSS_{i+1} = \SSS_i\setminus \{S_i\}$.

    If we did not conclude that $(G,k)$ is no-instance in the above procedure, we let $S = \bigcup_{P\in \SSS_q}V(P)$ be the set of all vertices of prisons in $\SSS_q$. Clearly if $G[S]\Delta A$ is prison-free, then $A$ intersects all prisons in $\SSS_q$. 
    We show that for every $A\subseteq E(G)$ with $|A|\le k$, $A$ intersects all sets in $\SSS_i$ if and only if it intersects all sets in $\SSS_{i+1}$. The lemma then follows by induction on $q$. 
    First note that $\SSS_{i+1}\subseteq \SSS_i=\SSS_{i+1}\cup \{S_i\}$. So we only need to show that if $A$ intersects all sets in $\SSS_{i+1}$, then it intersects also $S_i$. Note that $\SSS'_i\setminus \{S_i\}\subseteq \SSS_{i+1}$ is a sunflower with at most $k+1$ sets that pairwise intersect precisely in the same set of edges $C$ that is non-empty. Since $A$ intersects each of these sets and $|A|\le k$, it follows that $A$ has to intersect $C$. But $C\subseteq S_i$ and so $A\cap S_i \neq \emptyset$.
\end{proof}

\fi 
For the rest of the section and of the proof, 
we let $S$ be the set computed by Lemma~\ref{lem:kernelizing_P_F}. It follows, as long as we keep $G[S]$ as the subgraph of the reduced instance, we only need to be concerned about the prisons that are created by removing some edge from $G$, as all the prisons that were in $G$ to start with are hit by a set $A$ of at most $A$ edges as long as $G[S]\Delta A$ is prison-free. 
Given the above, the following two reduction rules are \iflong rather \fi straightforward.

\iflong 
\begin{reductionRule}
\fi
\ifshort  
\begin{reductionRule}[$\star$]
\fi
\label{red:notSuperEdge}
    If an edge $e\in (E(G) \setminus E(G[S]))$ is not in a strict supergraph of a prison, delete it. 
\end{reductionRule}

\begin{lemma}
  Reduction Rule~\ref{red:notSuperEdge} is safe. 
\end{lemma}
\iflong 
\begin{proof}
  

    Let $G'$ be the graph obtained by applying the rule on an edge $e$. 
    Assume $(G,k)$ is a yes instance, and let $A\subseteq V(G)$ be a minimal solution. $G\Delta A$ has no prisons. Assume that $G'\Delta A$ has an induced prison. It must contain the vertices of $e$, so $G\Delta A$ has a strict supergraph of a prison containing $e$, a contradiction.

    On the other hand, suppose that $(G',k)$ is a yes instance, and let $A\subseteq V(G')$ be a minimal solution of size at most $k$. If $G\Delta A$ has an induced prison $P$, then this prison contains $e$. Therefore $P$ is a prison of $G$ (it cannot be a strict supergraph of a prison by the assumption on $e$). Note that $G[S]=G'[S]$ and so $A\cap E(G[S])$ intersects every prison that is fully contained in $G[S]$ (as none of these prisons contain $e$). Therefore, by Lemma~\ref{lem:kernelizing_P_F}, $A$ intersects every prison in $G$. Hence, $A$ already contains an edge of $P$, which contradicts the assumption that $G\Delta A$ still contains $P$. 

    This concludes the proof by showing that $(G,k)$ is a yes instance if and only if $(G',k)$ is.
\end{proof}
\fi 

\iflong 
\begin{reductionRule}
\fi
\ifshort  
\begin{reductionRule}[$\star$]
\fi
\label{red:singleEdge}
 For every prison $P$ in $G$, if $E(G[S])\cap P = \{e\}$, then remove $e$ and decrease $k$ by one.
\end{reductionRule}

\begin{lemma}
  Reduction Rule~\ref{red:singleEdge} is safe.
\end{lemma}
\iflong 
\begin{proof}
    It follows from Lemma~\ref{lem:kernelizing_P_F} that if $A\subseteq E(G)$ with $|A|\le k$ such that $G[S]\Delta A$ is prison-free, then $A$ intersects every prison in $G$. Note that for $A' = A\cap E(G[S])$, $G[S]\Delta A'$ is also prison-free. Hence $A'$ intersects $P$. Since $e$ is the only edge of $P$ in $E(G[S])$, it follows $e\in A'\subseteq A$. It follows that $e$ is in every solution of size at most $k$ and we can safely remove it and decrease $k$ by one.
\end{proof}
\fi 

Thanks to these Reduction Rules, we found a set $S$ of vertices of size polynomial in $k$ such that for any subset $S'$ of vertices of $S$ such that $G[S']$ has at most one edge, $G[\bar{S} \cup S']$ is prison-free. We note that we assume that all reduction rules are applied exhaustively; that is whenever a reduction rule is applicable, we apply it and restart the process from the beginning. Hence, in all statements in the rest of the section, we implicitely assume that none of the reduction rules can be applied.

\subsection{Consequences on the Structure of $G[\bar{S}]$}

Now we are able to show the following properties that will be useful for the kernel. 
The following lemma gives us a stronger property than just the characterisation of prison-free graphs.

\iflong 
\begin{lemma}
\fi
\ifshort  
\begin{lemma}[$\star$]
\fi\label{lem:tripartite_neighbour}
    Let $F\in cmd_3(G[\bar{S}])$ be a maximal complete multipartite subgraph of $G[\bar{S}]$ such that $F$ has exactly three classes. Then there exists $s\in S$ with $F\subseteq N(s)$. Moreover, there is a strict supergraph of a prison that contains $s$ and an edge in $F$.
\end{lemma}
\ifshort
\begin{proof}[Proof Sketch]
    Since $F$ has three classes, there are $u,v,w$ such that $uvw$ is a triangle in $G[\bar{S}]$. 
    Due to Reduction Rule~\ref{red:notSuperEdge}, each edge of $F$ is in a strict supergraph of a prison and hence in $K_4$ in $G$. Now given a $K_4$ $(a,b,u,v)$ that contains $uv$, we conclude that $w$ is adjacent to either $a$ or $b$, else $(a,b,u,v,w)$ induces a prison with at most one edge in $S$ and Reduction Rule~\ref{red:singleEdge} applies. Hence $\{u,v,w\}$ is a subset of a $K_4$, say $(u,v,w,s)$. As a consequence of Lemma~\ref{lem:kcomp} and since $G[\bar{S}]$ is prison-free, one can show that $F$ has to be fully included in the $K_4$-free part of $G[\bar{S}]$, and hence $s\in S$. Now, any vertex $x\in F\setminus \{u,v,w\}$ is adjacent to exactly two vertices in $\{u,v,w\}$ and hence if $sx\notin E(G)$, then $(u,v,w,s,x)$ induces a prison without an edge in $S$, which is impossible due to the construction of $S$.
\end{proof}
\fi

\iflong
\begin{proof}
    Since $F$ has three classes, there are $u,v,w$ such that $uvw$ is a triangle in $G[\bar{S}]$. Let us first show that there is $s\in S$ such that $\{u,v,w\}\subseteq N(s)$, and then show that this implies that also $F\subseteq N(s)$.
    
    By Lemma~\ref{lem:kcomp}, every $K_4$ in $G[\bar{S}]$ is included in some $F'\in cmd_4(G[\bar{S}])$, moreover every vertex in $V(G[\bar{S}])\setminus F'$ has neighbors in at most one class of $F'$. Therefore, if a triangle in $F$ is in a $K_4$, then $F\subseteq F'$ for some $F'\in cmd_4(G[\bar{S}])$, which contradicts the fact that $F$ is maximal complete multipartite subgraph with only three classes. Hence, we can assume that no triangle in $F$ is in a $K_4$ in $G[\bar{S}]$.
    By Reduction Rule \ref{red:notSuperEdge} we know that $uv$ is in a strict 
  supergraph of a prison. Let $P = (u,v,a,b,c)$ be the strict supergraph of a prison. If $w\in \{a,b,c\}$, say without loss of generality $w=a$, then since $P$ contains at most one non-edge, either $b$ or $c$ forms a $K_4$ with $u,v,w$. Since $u,v,w$ are not in a $K_4$ in $G[\bar{S}]$, it follows that there is a vertex in $S$ that is adjacent to the triangle $(u,v,w)$ and that is a strict supergraph of a prison $P$ that contains the edge $uv$ in $F$. On the other hand, if $w\notin \{a,b,c\}$, then again, we can assume without loss of generality that $(u,v,a,b)$ is a $K_4$. Moreover, $w$ is adjacent to $u$ and $v$. By Reduction Rule~\ref{red:singleEdge} $(u,v,w,a,b)$ cannot be a prison, as it would contain only a single edge in $S$. Hence $w$ is adjacent either to $a$ or $b$, say $a$. Then, $(u,v,w,a)$ is a $K_4$ and $a\in S$, moreover, $a$ is in a strict supergraph of a prison $(u,v,w,a,b)$ that contains the edge $uv$ in $F$. And again we found a vertex $s\in S$ such that $\{u,v,w\}\subseteq N(s)$.

  Let us now show that if for $s\in N(s)$ and a triangle $(u,v,w)$ in $F$ we have $\{u,v,w\}\subseteq N(s)$, then $F\subseteq N(s)$. Let $x\in F\setminus \{u,v,w\}$. Clearly $x$ is adjacent to exactly two out of three vertices in $\{u,v,w\}$, as $F$ is complete multipartite graph with three classes. Moreover, $(u,v,w,x,s)$ cannot be a prison, by the construction of $S$, since it is edge-disjoint from $G[S]$. It follows that  $(u,v,w,x,s)$ is a strict supergraph of a prison, with only non-edge between $x$ and one vertex of $u,v, w$ and $xs\in E(G)$. 
\end{proof}
\fi 
\iflong 
\begin{lemma}
\fi
\ifshort  
\begin{lemma}[$\star$]
\fi\label{lem:superprisonfree} 
    For all $F\in cmd_3(G[\bar{S}])$, for all $x\in \bar{S} \setminus F$, $N_G(x)$ intersects at most one class of $F$.
\end{lemma}\ifshort
\begin{proof}[Proof Sketch]
    The Lemma follows for $F\in cmd_4(G[\bar{S}])$ by Lemma~\ref{lem:kcomp}. Hence, we can assume $F\in cmd_3(G[\bar{S}])\setminus cmd_4(G[\bar{S}])$. For the sake of contradiction, assume that for $x\in \bar{S}\setminus F$, we have that $N(x)\cap F$ contains an edge $uv$. Since $F$ has three classes, for every $w$ in the class $C$ that does not contain $u$ nor $v$, we have that $uvw$ is a triangle. Moreover, similarly as in the proof of Lemma~\ref{lem:tripartite_neighbour}, one can show using Lemma~\ref{lem:kcomp} that $(u,v,w,x)$ cannot be $K_4$ and so $wx\notin E(G)$. Since, $F$ is inclusion maximal, there is $a\in F\setminus C$ such that $xa\notin E(G)$. 
    By Lemma~\ref{lem:tripartite_neighbour}, there is $s\in S$ with $F\subseteq N(s)$. But then either $(u,v,w,s,x)$ or $(u,v,a,s,x)$ induces a prison with no edge in $G[S]$ (depending on whether $sx\in E(G)$ or not), which is impossible due to the construction of $S$. 
\end{proof}
\fi
\iflong
\begin{proof}
For all $F\in cmd_4(G[\bar{S}])$, the lemma follows from Lemma~\ref{lem:kcomp}, since $G[\bar{S}]$ is prison-free.  We will prove it for $F\in cmd_3(G[\bar{S}])$ with exactly three classes. 

Note that by Lemma~\ref{lem:kcomp}, 
 every $K_4$ in $G[\bar{S}]$ is included in some $F'\in cmd_4(G[\bar{S}])$, moreover every vertex in $V(G[\bar{S}])\setminus F'$ has neighbors in at most one class of $F'$. Hence, similarly as in the proof of Lemma~\ref{lem:tripartite_neighbour}, we can assume that no triangle in $F$ is in a $K_4$ in $G[\bar{S}]$. 
 Assume for a contradiction that for some $x\in \bar{S} \setminus F$, the neighborhood of $x$ intersects two classes of $F$. That is an edge $uv$ in $F$ such that $\{u,v\}\subseteq N_G(x)$. Moreover, since $F$ has three classes, there is $w$ such that $uvw$ is a triangle in $G[\bar{S}]$. 
By Lemma~\ref{lem:tripartite_neighbour}, there is $s\in S$ such that $F\subseteq N(s)$ and  $(s,u,v,w)$ is a $K_4$. Moreover, $(s,u,v,w,x)$ cannot be a prison, as it is edge-disjoint from $G[S]$ and so either $sx\in E(G)$ or $wx\in E(G)$.
  If $wx\in E(G)$, then $(x,u,v,w)$ induces $K_4$ in $G[\bar{S}]$, which, as we already argued, is impossible. Hence, $xs\in E(G)$. Note that $x$ cannot be adjacent to any vertex $w'$ in the same class as $w$, as $(x,u,v,w')$ would be a $K_4$ in $G[\bar{S}]$. Since $F$ is inclusion maximal, it follows that there is a vertex $a$ in the class of $u$ or in a class of $v$ that is not adjacent with $x$. Without loss of generality assume $a$ is in the class of $u$. Since $F\subseteq N(s)$, $sa$ is an edge. However, then $(w,v,a,s,x)$ 
  is a prison, which contradicts construction of $S$, since it is edge-disjoint from $G[S]$.  
\end{proof}
\fi 

\iflong 
\begin{lemma}
\fi
\ifshort  
\begin{lemma}[$\star$]
\fi\label{lem:noCrossPrison3}
    If a supergraph $P$ of a prison of $G$ has an edge in $F\in cmd_3(G[\bar{S}])$, then $P\subseteq S\cup F$.
\end{lemma}
\ifshort
\begin{proof}[Proof Sketch]
    For the sake of contradiction, let $P$ be a supergraph of a prison of $G$ that has an edge $uv$ in $F\in cmd_3(G[\bar{S}])$ and there is $x\in P\setminus (S\cup F)$.
    One can show that due to Lemma~\ref{lem:superprisonfree} and application of Reduction Rule~\ref{red:singleEdge}, $P = (a,b,u,v,x)$ is a strict supergraph of a prison with $\{a,b\}\subseteq S$ and $vx$ (or $ux$) being the only non-edge of $P$. Moreover, $F$ contains triangle $uvw$, as it has at least three classes and $wx\notin E(G)$ by Lemma~\ref{lem:superprisonfree}. Due to Reduction Rule~\ref{red:singleEdge},  $(u,v,w,a,b)$ cannot induce a prison, and either $wa\in E(G)$ or $wb\in E(G)$. If $wa\in E(G)$ (resp. $wb\in E(G)$), then $(x,a,u,w,v)$ (resp. $(x,b,u,w,v)$) induces a prison in $G$ with at most one edge in $G[S]$, contradicting application of Reduction Rule~\ref{red:singleEdge}.
\end{proof}
\fi

\iflong
\begin{proof}
    Let $P$ be a supergraph of a prison of $G$ that has an edge $uv$ in $F\in cmd_3(G[\bar{S}])$.
    
    Assume that there is $x\in P$ neither in $S$ nor $F$. Additionally, note that $P$ contains at most two vertices and hence at most one edge in $S$. It follows from construction of $S$ and Reduction Rule~\ref{red:singleEdge} that $P$ is a strict supergraph of a prison. 
    Since $x\in \bar{S}\setminus F$, $\{u,v\} \not\subseteq N_G(x)$ by Lemma \ref{lem:superprisonfree}. Without loss of generality $vx$ is the only non-edge between vertices of $P$, and $u$ is a neighbor of $x$. If a third vertex of $P$ is in $F$, $x$ would see more than one class from $F$ which contradicts Lemma~\ref{lem:superprisonfree}. If a second vertex is out of $F\cup S$, it would see the classes of $u$ and $v$ which again contradicts Lemma~\ref{lem:superprisonfree}. So the other two vertices of $P$, noted $a,b$, are in $S$.

     Note that $abuv$ is isomorphic to a $K_4$. Let $w\in F$ be in a class of $F$ different than $u$ and $v$. Since $(w,u,v,a,b)$ contains only one edge in $S$, it cannot be isomorphic to a prison. Hence $w$ is adjacent to either $a$ or $b$. Without loss of generality we can assume that $w$ is adjacent to $a$. But then $(x,v,a,u,w)$ is a prison that is edge-disjoint from $G[S]$, which contradicts the construction of $S$. Therefore, there is no vertex $x\in P$ outside of $S\cup F$ and the lemma follows.
\end{proof}
\fi 

Let's now focus on the edges of $\bar{S}$ that are not in any triangle. We show that since Reduction Rules~\ref{red:notSuperEdge}~and~\ref{red:singleEdge} has been exhaustively applied, even these edges can be partitioned to maximal complete bipartite subgraphs. 

Let $B$ be the set of edges of $G[\bar{S}]$ that are not in any triangle in $G[\bar{S}]$. For all $e=ab\in E(G[S])$, we note $B_e=B\cap \{uv\colon u,v\in N(a)\cap N(b)\}$. Note that every edge $f\in B$ belongs to a $K_4$ in $G$ due to the application of Reduction Rule~\ref{red:notSuperEdge}. Since, $f$ is not in any triangle in $G[\bar{S}]$, it follows that it is in some $B_e$ for $e\in E(G[S])$. On the other hand, we show that if $B_{e_1}\cap B_{e_2} \neq \emptyset$, then $B_{e_1} = B_{e_2}$, otherwise $G$ contains a prison with at most one edge in $G[S]$, which gives us the following lemma. 

\iflong 
\begin{lemma}
\fi
\ifshort  
\begin{lemma}[$\star$]
\fi\label{lem:weakCrosscmc2}
$\{B_e\mid e\in E(G[S])\}$ is a partition of $B$.
\end{lemma}
\iflong
\begin{proof}
    Let $uv\in B$. This edge is in a $K_4$ $uvab$ by Reduction Rule \ref{red:notSuperEdge}. It is not in a triangle of $G[\bar{S}]$ so $ab$ is an edge of $G[S]$ that has $uv$ as a common neighbor, e.g. $uv\in N_G(ab)$. So every element of $B$ is in the common neighborhood of an edge of $G[S]$.

    Let $e_1,e_2 \in E(G[S])$, noted $e_1=a_1b_1, e_2=a_2b_2$. Assume $B_{e_1} \neq B_{e_2}$ but there is a common edge $uv$ and an edge $xy$ that is not in both $B_{e_1}$ and $B_{e_2}$. Without loss of generality, we can say that the edge $xy$ is in $B_{e_2}$ and is not in $B_{e_1}$. It follows that at least one of vertices among $x,y$, say $x$, is not endpoint of the edge $uv$, i.e., $u\neq x\neq v$, and at the same time $x$ is not adjacent to one of the vertices $a_1$ or $b_1$, say $x$ is not adjacent to $a_1$ and there
    is no edge $xa_1$ (if $x$ is adjacent to both $a_1$ and $b_1$, then $y$ is not adjacent to one of them). Since $(x,a_2,b_2,u,v)$ is not a prison, there is an edge between $x$ and $uv$, we can assume without loss of generality that there is edge $xv$. Since $(a_1,u,v,a_2,b_2)$ is also not a prison, there is an edge $a_1a_2$ or $a_1b_2$, say $a_1b_2$. Moreover, note that $uv$ is not in any triangle in $G[\bar{S}]$, hence $xu$ is not an edge. It follows that $(x,b_2,v,u,a_1)$ is a prison which contradicts that Reduction Rule~\ref{red:singleEdge} has been applied exhaustively and the lemma follows.
\end{proof}
\fi 
The following lemma is a straightforward consequence of Reduction Rule~\ref{red:singleEdge}, since for every $e\in E(S)$, $G[\bar{S}\cup e]$ is prison-free and hence $N_G(e)\cap \bar{S}$ is $\bar{P_3}$-free.

\iflong 
\begin{lemma}
\fi
\ifshort  
\begin{lemma}[$\star$]
\fi\label{lem:bipart}
    For all $e\in E(G[S])$, $N_G(e)\cap \bar{S}$ is complete multipartite and if $B_e$ is non empty, $N_G(e)$ is a complete bipartite graph.
\end{lemma}
\iflong
\begin{proof}
    $N_G(e)\cap \bar{S}$ is $\bar{P_3}$-free otherwise there would be a prison that intersects $S$ only in $e$. This means that it is complete multipartite of width $\geq 0$. Furthermore, 
    if $B_e$ is non empty, there is a triangle-free edge in the complete multipartite component so it is complete bipartite.
\end{proof}
\fi 

\iflong 
\begin{lemma}
\fi
\ifshort  
\begin{lemma}
[$\star$]\fi 
\label{lem:noCrossPrison2}
    Let $e\in E(G[S])$. Assume that $B_e$ is not empty. Then any induced supergraphs of a prison $P$ that has an edge in $B_e$ is in $B_e \cup S$. 
\end{lemma}
\ifshort
\begin{proof}[Proof Sketch]
   Assume that there is an induced supergraphs of a prison $P = (a,b,u,v,w)$, where $uv$ in $B_e$ and $w\in \bar{S} \setminus B_e$. It follows from Reduction Rule~\ref{red:singleEdge} and the fact that $uv$ is not in a triangle in $G[\bar{S}]$ that (1) $P$ is a strict supergraph of a prison with only non-edge $uw$ (resp. $vw$) and (2) $\{a,b\}\subseteq S$. By Lemma~\ref{lem:weakCrosscmc2}, $B_e = B_{ab}$. Since $B_e$ is maximal complete bipartite subgraph of $G[\bar{S}]$, the class that contains $v$ (resp. $u$), contains a vertex $v'$ (resp. $u'$) with $v'w\notin E(G)$. But then $(a,b,v,v',w)$ (resp. $(a,b,u,u',w)$) induces a prison in $G$ with only one edge inside $S$, which is a contradiction.
\end{proof}
\fi
\iflong
\begin{proof}
    Assume that there is an induced supergraphs of a prison $P$ that has an edge $uv$ in $B_e$ and a vertex $w$ in $\bar{S} \setminus B_e$. Since $P$ contains at most two vertices and hence at most one edge in $S$, by the construction of $S$ and Reduction Rule~\ref{red:singleEdge}, $P$ is a strict supergraph of a prison and so is missing only a single edge. Since  
    $uv$ is not in a triangle of $G[\bar{S}]$, it follows that the only non-edge in $P$ is either $uw$ or $vw$ and $u,v,w$ are the only vertices of $P$ that are not in $S$. Without loss of generality $uw$ is not an edge of $G$. Let $a,b$ be the two remaining vertices of $P$. The edge $ab$ of $G[S]$ has $u,v$ and $w$ in its common neighborhood, but $vw$ is not an edge of $B_e=B_{ab}$. This means that $vw$ is not in $B$ and therefore occurs in a triangle $vwx$ where $x \in \bar{S} \setminus \{u,v,w\}$. Since Reduction Rule~\ref{red:singleEdge} has been applied exhaustively, $(x,v,w,a,b)$ is not a prison and so, without loss of generality, $xa$ is an edge. Similarly, $(u,v,a,x,w)$ cannot be a prison because it is edge-disjoint from $G[S]$. But, we already argued that $(v,x,w,a)$ is a $K_4$, moreover, $u$ is adjacent to $a$ and $v$, but non-adjacent to $w$ and $x$, as $uv$ is not in any triangle in $G[\bar{S}]$. Therefore, $(u,v,a,x,w)$ is a prison, which is a contradiction. 
\end{proof}
\fi 

It follows that we can partition the edges of $G[\bar{S}]$ in complete multipartite subgraphs, where one of these complete multipartite subgraphs can intersect another in at most one class. The following reduction rule lets us reduce the number of complete multipartite subgraphs in this partition to $|S|^3+|S|^2 = \mathcal{O}(k^{24})$. Given this bound, we will reduce size of each of these components as well as number of isolated vertices in $G[\bar{S}]$ by a polynomial function in $k$ as well.

\iflong 
\begin{reductionRule}
\fi 
\ifshort
\begin{reductionRule}[$\star$]
\fi \label{red:buble}
    Let $F'\in cmd_3(G[\bar{S}])$. If $F\subseteq V(G)$ is a maximal complete multipartite component such that $F\supseteq F'$ and for every vertex $v \notin F$, $N(v)$ intersects at most one class of $F$, remove all edges of $F$. 
\end{reductionRule}

\begin{lemma}
  Reduction Rule~\ref{red:buble} is safe.
\end{lemma}
\ifshort
\begin{proof}[Proof Sketch]
    It is not difficult to show that there is no supergraph of a prison in $G$ that contains an edge with both endpoints in $F$ and at the same time an edge with at least one endpoint outside of $F$. Given this for every $A\subseteq E(G)$, every prison in $G\Delta A$ is either all edge in $F$ or all edges in $G - E(F)$, since $F$ is complete multipartite and hence prison-free, it suffices to hit all prisons in $G-E(F)$. 
\end{proof}
\fi
\iflong
\begin{proof}
    Let $G'$ be the graph obtained from $G$ by removing all edges of $F$. 

    The safeness of the reduction rule follows from the fact that there is no supergraph of a prison in $G$ that contains an edge with both endpoints in $F$ and at the same time an edge with at least one endpoint outside of $F$. Given this for every $A\subseteq E(G)$, every prison in $G\Delta A$ is either all edge in $F$ or all edges in $G'$, since $F$ is complete multipartite and hence prison-free, it suffices to hit all prisons in $G'$. 

    To see that no prison can have an edge in $F$ and outside of $F$, let us assume for contradiction that $P$ is a supergraph of a prison with vertices $a,b,c,d,e$ such that $(a,b,c,d)$ is $K_4$ and $e$ is adjacent to $c,d$. If any of the edges in $(a,b,c,d)$ is in $F$, then $(a,b,c,d)$ is fully in $F$ as the remaining two vertices would see two classes of $F$. But then $e$ also sees two classes of $F$ (the class of $c$ and the class of $d$) and so $e$ and the edges $ed$ and $ec$ are also in $F$. It follows that all edges of $P$ are in $F$. On the other hand, if one of the edges $ec$ or $ed$, say $ec$ is in $F$, then $d$ would see two classes of $F$ and so $cd$ would be in $F$ as well and the same argument shows that all edges of $P$ are again in $F$. Therefore any supergraph of a prison that contains at least one edge in $F$ is fully contained in $F$. Since $F$ is prison-free, it follows that the reduction rule is safe.
\end{proof}
\fi 

Recall that we always assume that none of the previous reduction rules can be applied, in particular from now on we assume also that Reduction Rule~\ref{red:buble} is not applicable.

From now on, we denote $\cmsG=cmd_3(G[\bar{S}]) \cup \{B_e\mid e\in E(G[S])\}.$ 

\iflong 
\begin{lemma}
\fi
\ifshort  
\begin{lemma}[$\star$]
\fi 
\label{lem:boundOnNbC}
 The edges of $G[\bar{S}]$ can be partitioned into at most $|S|^3+|S|^2$ many maximal complete multipartite subgraphs of $G[\bar{S}]$. 
\end{lemma}
\iflong
\begin{proof}
\fi
\ifshort
\begin{proof}[Proof Sketch]
\fi
Clearly every edge in $G[\bar{S}]$ is in a maximal complete multipartite subgraph of $G[\bar{S}]$\iflong, as $K_2$ is a complete multipartite graph\fi. We only need to show that each edge is in precisely one such subgraph and that their number is at most $|S|^3+|S|^2$. First note the edges in $B$\iflong,i.e., edges of $G[\bar{S}]$ that are not in a triangle of $G[\bar{S}]$,\fi{} are partitioned into at most $|E(G[S])|\le |S|^2$ many complete bipartite graphs by Lemma~\ref{lem:weakCrosscmc2}. 
Hence, we only need to consider the edges that are \iflong in a triangle of $G[\bar{S}]$, i.e., it is an edge of \fi in some $F\in cmd_3(G[\bar{S}])$. 
\ifshort
It is a rather straightforward consequence of Lemma~\ref{lem:superprisonfree} that any edge $e$ in $F$ cannot be in any other maximal complete multipartite subgraph in $cmd_3(G[\bar{S}])$. 
\fi
\iflong 
Note that for $F\in cmd_3(G[\bar{S}])$, every vertex $x\in \bar{S}\setminus F$ has neighbors in at most one class of $F$ by Lemma~\ref{lem:superprisonfree}, hence $x$ cannot be adjacent to any edge in $F$. Therefore, an edge $e$ in $F$ cannot be in a complete multipartite subgraph of $G[\bar{S}]$ that contains a vertex outside of $F$. Since $F$ is maximal, edges that are in $F$ cannot be in any other maximal complete multipartite subgraph in $G[\bar{S}]$ and elements of $cmd_3(G[\bar{S}])$ are pairwise edge-disjoint.
\fi 
It remains to show that show that $|cmd_3(G[\bar{S}])|\le |S|^3$. 

We will now define a function $f:cmd_3(G[\bar{S}]) \rightarrow cmd_4(G)$ such that for all $F\in cmd_3(G[\bar{S}])$, it holds that (1) $F\subseteq f(F)$ and (2) $f(F)$ contains at least one vertex in $S$. That is $f$ is injective. We will show that every edge in $G[S]$ is in at most $|S|$ many complete multipartite subgraph in the image of $f$ and every element of the image of $f$ contains an edge of $G[S]$, bounding $cmd_3(G[\bar{S}])$ by $|E(G[S])|\le |S|^3$. 

If $F\in cmd_4(G[\bar{S}])$, then \ifshort since Reduction Rule \ref{red:buble} has been exhaustively applied, there is $v\in S$ such that $N(v)$ intersects at least two classes of $G[F]$, and since $G[\bar{S}\cup\{v\}]$ is prison-free and $F$ contains a $K_4$, $G[F\cup\{v\}]$ is complete multipartite. \fi \iflong if for all $v\in S$, $N(v)$ intersects at most one class of $G[F]$, Reduction Rule \ref{red:buble} can be applied on $F$. So we can assume that there is $v\in S$ such that $N(v)$ intersects at least two classes of $G[F]$, and since $G[\bar{S}\cup\{v\}]$ is prison-free and $F$ contains a $K_4$, $G[F\cup\{v\}]$ is complete multipartite. \fi We define $f(F)$ as any maximal complete multipartite component of $G$ containing $F\cup\{v\}$. 
Else, if $F\in cmd_3(G[\bar{S}])-cmd_4(G[\bar{S}])$, then by Lemma~\ref{lem:tripartite_neighbour}, there is a vertex $s\in S$ such $F\subseteq N(s)$ and $G[F\cup\{s\}]$ is a complete multipartite graphs with at least four parts\iflong{} (importantly, it contains $K_4$)\fi. We define $f(F)$ as a maximal complete multipartite subgraph of $G$ containing $F\cup\{s\}$. 

\iflong 
Note that in both cases $f(F)\subseteq F\cup S$, because any subgraph of a complete multipartite graph is also complete multipartite and $F$ is already inclusion maximal in $G[\bar{S}]$. 
We thus defined a function $f:cmd_3(G[\bar{S}]) \rightarrow cmd_4(G)$ that is injective, and for all $F$, $f(F)$ has at least one vertex in $S$.

\fi 
\iflong
\begin{claim}
\fi
\ifshort
\begin{claim}[$\star$]
\fi
    Let $e=uv$ be an edge of $G[S]$. Then at most $|S|$ elements of $Im(f)$ contains the vertices of $e$.
\end{claim}
\iflong
\begin{claimproof}  
    For the sake of contradiction, assume that there is $F_1,F_2\in cmd_3(G[\bar{S}])$ such that $F_1\neq F_2$ and $u,v \in f(F_1) \cap f(F_2)$. $G[F_1\cup F_2 \cup \{u,v\}]$ is prison-free due to Reduction Rule~\ref{red:singleEdge}.
    If $F_1\cup \{u,v\}$ (or $F_2\cup \{u,v\}$) contains $K_4$, then let $F'$ be a maximal complete multipartite subgraph of $G[F_1\cup F_2 \cup \{u,v\}]$ that contains $F_1\cup \{u,v\}$ (or $F_2\cup \{u,v\}$). Then every vertex $x\in F_2\setminus F'$ (or $x\in F_1\setminus F'$) is adjacent to the edge $uv$ and hence to at least two classes of $F'$. By Theorem~\ref{the:handbag}, no such vertex $x$ exists and $F' = F_1\cup F_2 \cup \{u,v\}$. But then $F_1\cup F_2$ induces a complete multipartite subgraph of $G[\bar{S}]$ and by maximality of $F_1$ and $F_2$, we have $F_1=F_2$. 
    Hence, if for some $F\in cmd_3(G[\bar{S}])$ such that $\{u,v\}\subseteq f(F)$, $F\cup \{u,v\}$ contains $K_4$, then $F$ is the unique element of $cmd_3(G[\bar{S}])$ with $\{u,v\}\subseteq f(F)$. 

    For the rest of the proof, let us assume that both $F_1\cup \{u,v\}$ and $F_2\cup \{u,v\}$ are $K_4$-free. By the definition of $f$, 
    there exist $s_1,s_2\in S$ such that \begin{itemize}
        \item $F_1\subseteq N(s_1)$ and $F_2\subseteq N(s_2)$ and 
        \item $s_1\in f(F_1)$ and $s_2\in f(F_2)$.
    \end{itemize} 
    We will show that $s_1\notin f(F_2)$ (and symmetrically $s_2\notin f(F_1)$). Let $x\in F_2\setminus F_1$, we show that $xs_1$ cannot be an edge. Since, $x$ is a vertex in $\bar{S}\setminus F_1$, by Lemma~\ref{lem:superprisonfree}, $N_G(x)$ intersect at most one class of $F_1$. 
    Let us distinguish two cases depending on whether $x$ has a neighbor in $F_1$ or not. If $x$ has a neighbor $a\in F_1$, then there is a triangle $(a,b,c)$ in $F_1$ such that $(a,b,c,s_1)$ is $K_4$ and $x$ is neighbor of precisely $a$ and $s_1$. Hence, $(a,b,c,s_1,x)$ is a prison without any edge in $S$, a contradiction to the construction of $S$, hence $x$ cannot be a neighbor of $s_1$.
    Now assume $N(x)\cap F_1=\emptyset$, notice that $u$ and $v$ are both in different classes of $f(F_1)$ than $s_1$, since neither of them form $K_4$ with $F_1$ and  $s_1u, s_1v\in E(G)$ and $f(F_1)$ contains  a $K_4$ $(s_1,u,a,b)$, where $a,b\in F_1$. If $xs_1\in E(G)$, then $(s_1,u,a,b,x)$ is a prison with a single edge in $S$, which is impossible due to Reduction Rule~\ref{red:singleEdge}. It follows that for every $x\in F_2\setminus F_1$ we have that $xs_1$ is not an edge. Since $F_2$ and $F_1$ do not share any edge, they can intersect in at most one class and $F_2$ contains at least two classes that are not adjacent to $s_1$, hence $s_1\notin f(F_2)$. Hence for every $F\in cmd_3(G[\bar{S}])\setminus cmd_4(G[\bar{S}])$ with $\{u,v\}\subseteq f(F)$ there is a unique $s_F\in f(F)$ with $F\subseteq N(s_F)$. Therefore, there are at most $|S|$ such $F\in cmd_3(G[\bar{S}])$.
\end{claimproof}
\fi

    There are thus at most $|S|^3$ elements of $Im(f)$ that contains an edge in $S$. 
\iflong
\begin{claim}
\fi
\ifshort
\begin{claim}[$\star$]
\fi
    Each element of $Im(f)$ contains an edge in $S$.
\end{claim}
\iflong 
 \begin{claimproof}
     Assume that there is $F\in Im(f)$ that does not. Since Reduction Rule \ref{red:buble} cannot be applied anymore, there is $x\notin F$ that sees at least two classes of $F$. Moreover, $F$ is inclusion maximal, so $G[F\cup \{x\}]$ is not complete multipatite subgraph. Since $F$ contains at least 4 parts, it follows from Theorem~\ref{the:handbag} that $G[F\cup \{x\}]$ contains a prison $P$. Moreover, $P-x$ is a complete multipartite graph. It follows that if $P=(a,b,c,d,e)$ such that $(a,b,c,d)$ is a $K_4$ and $e$ is adjacent to $c$ and $d$, then $x\in \{a,b,e\}$. In particular $N_G(x)\cap P$ is either $K_2$ or $K_3$.  Since Reduction Rule~\ref{red:singleEdge} is not applicable, $G[P]$ contains at least two edges in $S$. By our assumption, all edges in $G[F\cup \{x\}]$ in $S$ are incident on $x$. However, as we argued $N_G(x)\cap P$ is complete, so if $G[P\cap S]$ contains two edges incident on $x$, say $xu$ and $xv$, then it also contains the edge $uv$. But $u,v \in F$, which is a contradiction.   
 \end{claimproof}   
\fi    

    So $|Im(f)| \leq |S|^3$. Since $f$ is injective, $|cmd_3(G[\bar{S}])| \leq |S|^3$. 
Therefore,  $|\cmsG| = |cmd_3(G[\bar{S}])| + |\{B_e\mid e\in E(G[S])\}| \leq |S|^3+|S|^2$.
\end{proof}

 We will have a kernel once we have bounded the size of a maximal complete multipartite subgraphs and the number of isolated vertices in $G[\bar{S}]$. Before we show how to bound those, let us observe the following two simple lemmas.

\iflong 
\begin{lemma}
\fi
\ifshort  
\begin{lemma}[$\star$]
\fi\label{prop:consistentFs3}
    Let $F\in cmd_3(G[\bar{S}])$ and $s\in S$.
     If $N(s)$ intersects more than one class of $F$, then either $N(s)\cap F = F\setminus C$, where $C$ is a single class in $F$, or $N(s)\cap F = F$.
\end{lemma}
\iflong
\begin{proof}
    
    If $F\in cmd_4(G[\bar{S}])$, the  statement follows from the fact that $G[F\cup \{s\}]$ is prison-free, the fact that $F$ induces a complete multipartite graphs that contains $K_4$ and from Theorem~\ref{the:handbag}.

    From now on assume that $F$ has exactly three classes.
    Due to Lemma~\ref{lem:tripartite_neighbour}, there is $s_1\in S$ such that $F\subseteq N(s_1)$. $G[F\cup\{s_1,s\}]$ is prison-free, because of Reduction Rule~\ref{red:singleEdge}. Moreover, $F\cup \{s_1\}$ induces a complete multipartite graph that contains a $K_4$. Since $N(s)$ intersects more than one parts of  $F\cup \{s_1\}$, it follows from Theorem~\ref{the:handbag} that $F\cup \{s_1,s\}$ is also complete multipartite and the statement follows.
\end{proof}
\fi 
\iflong 
\begin{lemma}
\fi
\ifshort  
\begin{lemma}[$\star$]
\fi\label{prop:consistentFs2}
    Let $e$ such that $B_e$ is non empty, and let $s\in S$. 
    $N(s)$ either contains $B_e$ or it intersects at most one class of $G[B_e]$.
\end{lemma}
\iflong
\begin{proof}
    Note that by Reduction Rule~\ref{red:notSuperEdge}, every edge $uv$ in $B_e$ is in a strict supergraph of a prison $P$. This supergraph of a prison contains a $K_4$ $(a,b,u,v)$ and by definition of $B_e$, $\{a,b\}\subseteq S$. By Lemma~\ref{lem:weakCrosscmc2}, $B_e=B_{ab}$.
    Assume that $N(s)$ contains an edge $uv$ of $B_e$.
    If $s\in \{a,b\}$, then $B_e\subseteq N(s)$ by definition of $B_e$. If $s\notin \{a,b\}$, then it follows from Reduction Rule~\ref{red:singleEdge} that either $sa$ or $sb$ is an edge in $G[S]$, else $(s,u,v,a,b)$ would be a prison with a single edge in $S$. Without loss of generality, assume that $sa\in E(G)$. Then $uv\in B_{sa}$ and by Lemma~\ref{lem:weakCrosscmc2} $B_e=B_{sa}$, so $B_e\subseteq N(s)$.
\end{proof}
\fi

\subsection{Marking important vertices}

We will now mark some important vertices that we will keep to preserve the solutions and afterwards argue that we can remove all the remaining vertices without changing the value of an optimal solution. To better understand why this marking procedure works, it is useful to recall Lemmas~\ref{lem:noCrossPrison3}~and~\ref{lem:noCrossPrison2} that state that any (not necessarily strict) supergraph of a prison that has at least one edge in some $F\in \cmsG$ is fully contained in $S\cup F$. That is only supergraphs of the prison that can contain vertices in more than one complete multipartite subgraph from $\cmsG$ are those with all edges either in $S$ or between $S$ and $N(S)$. Such supergraphs of a prison have always at most two vertices outside of $S$. 
Before we give the details of the marking procedure, let us introduce a concept of \emph{neighborhood pattern}. Let $X\subseteq V(G)$ and $a\in V(G)\setminus X$, then the neighborhood pattern of $a$ in $X$ is $N(a)\cap X$. Moreover, for two vertices $a,b\in V(G)\setminus X$, we say $a$ and $b$ have the same \emph{neighborhood pattern} in $X$ if $N(a)\cap X = N(b)\cap X$.

Now, let $F\in \cmsG$. By Lemmas~\ref{prop:consistentFs3}~and~\ref{prop:consistentFs2}, for every $s\in S$, if $N(s)\cap F \neq \emptyset$, then there are only three possibities how $N(s)$ and $F$ can interact. Namely,
\begin{enumerate}
    \item $N(s)\cap F\subseteq C$ for a single part $C$ of $F$;
    \item $N(s)\cap F = F\setminus C$ for a single part $C$ of $F$;
    \item $N(s)\cap F = F$.
    \end{enumerate}
Given the above, we can split the classes $C$ of $F$ into two types. 
\begin{description}
    \item[Type 1.] There is $s\in S$ such that $N(s)\cap F\subseteq C$ or $N(s)\cap F = F\setminus C$. We denote the set of classes of Type 1. as 
    $\typeOne{F}$;
    \item[Type 2.] The rest, which we denote $\typeThree{F}$.
\end{description}

Note that $|\typeOne{F}|\le |S|$ due to Lemmas~\ref{prop:consistentFs3}~and~\ref{prop:consistentFs2} all classes $C$ in $F$ with $C\notin \typeOne{F}$ have exactly the same neighborhood in $S$.

We compute a set $M_F$ of marked vertices for the component $F\in \cmsG$ as follows. We note that whenever we say, we mark some number $x$ of vertices with some property, we mean that if there are more than $x$ many vertices with that property, we mark arbitrary $x$ many of them, else we mark all of them.
Let us start with marking vertices in a class $C$ of Type 1. for every $C\in \typeOne{F}$.
For each $S'\subseteq S$ with $|S'|= 4$, and for each neighborhood pattern $\xi$ in $S'$, we add to $M_F$ arbitrary $2k+5$ vertices of $C$ with the neighborhood pattern $\xi$ in $S'$. Observe that for any $S''$ with $|S''| < 4$, there is $S'\supset S''$ with $|S'|=4$ and so for any neighborhood pattern $\xi''$ in $S''$, there is a neighborhood pattern in $S'$ that is equal to $\xi''$ if restricted to $S''$. Hence, using this marking, we marked at least $2k+5$ vertices with any neighborhood pattern in any $S'$ with $|S'|\le 4$ as well.

In addition, we pick arbitrary $2k+5$ classes of $F$ that are in $\typeThree{F}$ and add arbitrary $2k+5$ vertices from each picked class to $M_F$. 

\iflong 
\begin{lemma}
\fi
\ifshort  
\begin{lemma}[$\star$]
\fi\label{lem:component_marking_bound}
    Given the above marking procedure, for all $F\in \cmsG$, it holds that $|M_F| \le |S|^5\cdot (2k+5) + (2k+5)^2$.
\end{lemma}
\iflong
\begin{proof}
It follows from Properties~\ref{prop:consistentFs3}~and~\ref{prop:consistentFs2} that
    $\typeOne{F}$ is well defined and $|\typeOne{F}|\le |S|$. 
    Now for each $C\in \typeOne{F}$, $M_F$ contains at most $2^4\cdot \binom{|S|}{4} \cdot (2k+5)\le 16\frac{|S|^4}{24} \cdot (2k+5)\le |S|^4\cdot (2k+5)$ vertices. Hence, there are at most $|S|^5\cdot (2k+5)$ vertices from $\typeOne{F}$ in $M_F$. 
    Finally, there are  
    at most $(2k+5)^2$ vertices from classes in $\typeThree{F}$. 
\end{proof}
\fi 
In addition to marking the set $M_F$ for each $F\in \cmsG$, we mark additional set of at most $2^4\cdot \binom{|S|}{4}\cdot (2k+5)$ vertices by going over all subsets $S'$ of $S$ of size $4$ and for each neighborhood pattern $\xi$ in $S'$, we mark at most $2k+5$ additional vertices of $V(G)\setminus S$ that are not in any $F\in \cmsG$ with the given neighborhood pattern $\xi$ in $S'$. We denote this set of at most $|S|^4\cdot (2k+5)$ vertices as $M_\emptyset$. Additionally, observe that this way, we mark at least $2k+5$ vertices for each neighborhood pattern in any subset of $S$ of size at most four as well. 

\iflong 
\begin{lemma}
\fi
\ifshort  
\begin{lemma}[$\star$]
\fi\label{lem:marking_correctness}
    Let $G' = G[S\cup M_\emptyset\cup \bigcup_{F\in \cmsG}M_F]$. Then $(G,k)$ is yes-instance of \delProblem if and only if $(G',k)$ is. In addition $|V(G')| =  \bigoh(k^{65})$. 
\end{lemma}
\begin{proof}
    The bound on the size of $G'$ follows from the fact that that $|S| = \bigoh(k^8)$, $|\cmsG| \le |S|^3 + |S|^2$,\ifshort{} and Lemma~\ref{lem:component_marking_bound}. \fi\iflong{} and the fact that marking procedure marked at most $|\cmsG|\cdot (|S|^5\cdot (2k+5) + (2k+5)^2) + |S|^4\cdot (2k+5)$ many vertices outside of $S$, which is $\bigoh((k^8)^3\cdot (k^8)^5\cdot k) = \bigoh(k^{65})$ many vertices outside of $S$. 
    
    \fi
    Now, $G'$ is an induced subgraph of $G$ and hence for every $A\subseteq E(G)$, if $G\Delta A$ is prison-free, then also $G'\Delta (A\cap E(G'))$ is. Therefore, if $(G,k)$ is yes-instance, then so is $(G',k)$. 

    For the rest of the proof, let us assume that $(G',k)$ is yes-instance and let $A\subseteq E(G')$ be such that $|A|\le k$ and $G'\Delta A$ is prison-free. We show that $G\Delta A$ is also prison-free. For the sake of contradiction, let's assume that $P = (a,b,c,d,e)$ induces a prison in $G\Delta A$. Let us distinguish two cases depending on whether $P$ contains an edge between two vertices outside of $S$ or not. 
    \begin{description}
        \item[Case 1.] All edges of $P$ have at least one endpoint in $S$. Note that in this case, there are at most two vertices of $P$ outside of $S$. 
        
        If only one vertex of $P$, say $a$, is outside of $S$, then clearly $a$ is the only vertex of $P$ not in $G'$. Note that this also means that none of the edges incident with $a$ is in $A$. Moreover, $G'$ contains at least $2k+5$ vertices with the same neighborhood pattern in $\{b,c,d,e\}$, else $a$ would have been either in $M_\emptyset$ or in $M_F$ for some $F\in \cmsG$. Since $|A|\le k$, it follows that at least one of these $2k+5$ vertices, let's call it $a'$, is not incident to an edge in $A$ and in particular to an edge between $a'$ and a vertex in $\{b,c,d,e\}$.
        However, then $(G'\Delta A)[\{a',b,c,d,e\}]$ is isomorphic to $P$ and induces a prison, which is a contradiction with $G'\Delta A$ being prison-free.

        Now assume that two vertices $a$ and $b$ are outside of $S$, and there is no edge between $a$ and $b$ in $G \Delta A$.
        Note that the other non-edge of $P$ share an endpoint with $ab$, and w.l.o.g., we can assume it is $ac$. As at least one of $a$ and $b$ is not in $G'$, $ab\notin E(G)$. Moreover, $ac\in E(G)\cap A$, else $P$ is a prison in $G$ and Lemma~\ref{lem:kernelizing_P_F} implies that $A$ intersects $P$. Hence, $b\notin V(G')$, by our marking procedure, there are at least $2k+5$ vertices $b'$ in $V(G')\setminus S$ with $\{c,d,e\}\subseteq N(b')$ and $bb'\notin E(G)$. One of these vertices is not incident on any edge in $A$. For such a vertex $b'$, either $(a,b',c,d,e)$ or $(a,b,b',d,e)$ is a prison in $G'\Delta A$ depending on whether $ab'\in E(G)$ or $ab'\notin E(G)$. 
        \iflong  Note that at least one $a$ or $b$ does not belong to $G'$, else $P$ is an induced prison in $G'\Delta A$ as well. It follows that $ab\notin E(G)$, since $ab\notin E(G \Delta A)$ and $ab\notin A$. Finally, the second non-edge of $P$ has to be incident on $a$ or $b$, since the two non-edges of a prison share a common endpoint. Without loss of generality, let us assume that the non-edge is $ac$. Note that if $ac\notin E(G)$, then $P$ is a prison in $G$ as well. However, $G[S]$ is an induced subgraph of $G'$ and hence $G[S]\Delta A$ is prison-free. 
        It then follows by Lemma~\ref{lem:kernelizing_P_F} that $P$ is also hit by $A$, contradicting the fact that $P$ is a prison in $G\Delta A$. Hence, we can assume that $ac\in A$ and hence $a \in V(G')$. Now, we can assume that $b\in V(G)\setminus V(G')$, else $P$ would be contained in $G'$.  
        Since $|A|\le k$, there are at most $2k$ vertices incident on edges in $A$. Hence, by construction of $M_\emptyset$ and of $M_F$ for all $F\in \cmsG$, there exists $b'\in V(G')$ such that (1) $\{c,d,e\}\subseteq N(b')$, (2) $bb'\notin E(G)$, and (3) $b'$ is not incident on any edge in $A$. Now, if $ab'\notin E(G)$, then $P' = (a,b',c,d,e)$ is a prison $G'\Delta A$ with the two non-edges $ac$ and $ab'$. 
        On the other hand, if $ab'\in E(G)$, then $(a,b,b',d,e)$ is a prison in $G$ with the two non-edges $ab$ and $bb'$ and with a single edge in $S$, which contradicts the assumption that Reduction Rule~\ref{red:singleEdge} has been applied exhaustively. This concludes the proof of the first case.
        \fi
        
        \item[Case 2.] $P$ contains an edge in $G-S$. Then by Lemmas~\ref{lem:noCrossPrison3}~and~\ref{lem:noCrossPrison2}, it follows that there exists $F\in \cmsG$ such that $P\subseteq S\cup F$. Let $S' = P\cap S$. Now, for $x\in P\setminus S$, If $x$ belongs to $C\in \typeOne{F}$, then by construction of $M_F$ and the fact that $|A|\le k$, $M_F$ either contains $x$ or a vertex $x'$ such that (1) $x'\in C$, (2) $N(x)\cap S' = N(x)\cap S'$, and (3) $x'$ is not incident to any edge of $A$. On the other hand, if some of the vertices in $P\cap F$ are in the classes in $\typeThree{F}$, then all such vertices have exactly same neighborhood in $S$, moreover, either all their respective classes contain vertices in $M_F$, or there are at least five classes of $F$ that each has $2k+5$ vertices in $G'$ and no vertex in these classes is incident on an edge in $A$. Hence, for each vertex $x$ in $P\cap F$ that belongs to a class $C\in \typeThree{F}$, we can find $C'\in \typeThree{F}$ and $x'\in C'$ such that (1) $x'\in M_F$, (2) $x'$ is not incident on any edge in $A$, and (3) if $x$ and $y$ are in $P\cap F$, then $x'$ and $y'$ are in the same class of $F$ if and only $x$ and $y$ are. Let $P'=(a',b',c',d',e')$ be a subgraph of $G'$ such that for all $x\in \{a,b,c,d,e\}$, if $x\in G'$, then $x'=x$ and else $x'$ is computed as described above, depending on whether $x\in \typeOne{F}$ or $x\in \typeThree{F}$. It follows that for all $x,y\in \{a,b,c,d,e\}$, there is an edge $xy\in E(G)$ if and only if $x'y'\in E(G')$. Moreover, since $xy\in A$ implies that $\{x,y\}\subseteq V(G')$, it follows that  $xy\in A$ if and only if $x'y'\in A$. Therefore, $P'$ induces a prison in $G'$, which is a contradiction. 
    \end{description}
    Hence, such prison $P$ in $G\Delta A$ cannot exist and $G\Delta A$ is prison-free as well. Consequently, $(G',k)$ is yes-instance of \delProblem\ if and only if $(G,k)$ is and the Lemma follows. 
\end{proof}
The polynomial kernel for \delProblem\ then follows from Lemma~\ref{lem:marking_correctness} by observing that all reduction rules as well as marking procedure can be applied in polynomial time. 
\thmkernel*
\iflong
\begin{proof}
    Let us summarize the whole kernel and argue that the running time is polynomial. The correctness and the size of the kernel then follows from Lemma~\ref{lem:marking_correctness}. 

    We start by running the algorithm of Lemma~\ref{lem:kernelizing_P_F} to find in polynomial time a set of vertices $S$ such that if some $A\subseteq E(G)$ with $|A|\le k$ hits all the prisons in $G[S]$, then it hits all the prisons in $G$. That is for any prison $P$ in $G$, if $A$ intersects all prisons in $G[S]$, then it intersects $P$. Afterwards we exhaustively apply Reduction Rules~\ref{red:notSuperEdge}~and~\ref{red:singleEdge}. Each of these two rules removes an edge when applied, so they can be applied at most $|E(G)|$ times before the instance becomes trivial. In addition, we can check whether any of the two rules apply by exhaustively enumerating all at most $|V(G)|^5$ subsets of 5 vertices. 
    By Lemma~\ref{lem:superprisonfree}, any edge of $G[\bar{S}]$ that belongs to a triangle in $G[\bar{S}]$ belongs to unique maximal complete multipartite subgraph of $G[\bar{S}]$. Similarly, any edge of $G[\bar{S}]$ that does not  belong to a triangle in $G[\bar{S}]$ belongs to a single maximal complete bipartite subgraph of $G[\bar{S}]$ by Lemmas~\ref{lem:weakCrosscmc2}~and~\ref{lem:bipart}. Hence, we can easily enumerate all maximal complete multipartite subgraphs of $G[\bar{S}]$ by starting from an edge and greedily adding vertices one by one if they form a complete multipartite graph with the vertices already picked. Now for a each maximal complete $C$ multipartite subgraphs of $G[\bar{S}]$ with at least three classes, we can easily verify whether we can apply Reduction Rule~\ref{red:buble} to it. We add to $C$ all the vertices of $S$ that neighbor vertices in at least two classes of $C$ to obtain $C'\subseteq C$. If $G[C']$ is not complete multipartite graph or there is $x\in V(G)\setminus C'$ with neighbors in two classes of $C'$, then Reduction Rule~\ref{red:buble} cannot be applied to any supergraph of $C$, else it can be applied to $C'$. Therefore, Reduction Rule~\ref{red:buble} can also be applied in polynomial time. 

    Assuming exhaustive application of all reduction rules, Lemma~\ref{lem:boundOnNbC} shows that we can assume that there are at most $|S|^3+|S|^2$ many maximal complete multipartite subgraphs of $G[\bar{S}]$ and they are all pairwise edge-disjoint. For each such subgraph $C$, we apply a marking procedure to obtain the set of vertices $M_C$. To compute $M_C$, we need to enumerate all subsets $S'$ of $S$ of size four and partition vertices of each class in $C$ depending on their neighborhood in $S'$, which can be easily done in polynomial time as well. Therefore, the kernelization algorithm can be performed in polynomial time and the correctness and the size bound follow from Lemma~\ref{lem:marking_correctness} as mentioned above. 
\end{proof}

\fi

\section{Conclusions} \label{sec:conc}

We have showed that \textsc{$H$-free Edge Deletion} has a polynomial kernel when $H$ is the 5-vertex graph we call the ``prison'' (consisting of $K_5$ minus two adjacent edges). On the other hand, \textsc{$H$-free Edge Completion} for the same graph $H$ does not have a polynomial kernel unless the polynomial hierarchy collapses. By edge complementation, this is equivalent to the statement that
\textsc{$\overline{H}$-free Edge Deletion} has no polynomial kernel, where $\overline{H}$ is the edge complement of $H$. 
The positive result refutes a conjecture by Marx and Sandeep~\cite{1}, who conjectured that \textsc{$H$-free Edge Deletion} has no polynomial kernel for any graph $H$ on at least five vertices except trivial cases. 

In~\cite{1}, the conjecture is reduced to the statement that \textsc{$H$-free Edge Deletion} has no polynomial kernel for any graph $H$ in a list $\mathcal{H}$ of nineteen small graphs, via a sequence of problem reductions. In this naming scheme, the prison is the complement of $H_1$. The exclusion of co-$H_1$ from this list introduces a sequence of new minimal graphs $H'$ for which the kernelization problem is open, out of which the smallest are the prison plus a vertex $v$ which is (respectively) an isolated vertex; attached to a degree-3 vertex of the prison; or attached to both a degree-3 and a degree-2 vertex of the prison 
(R.~B.~Sandeep, personal communication, using software published along with~\cite{1}).
It is at the moment not known whether the new list $\mathcal{H}'$ is finite using the methods of~\cite{1}.

More broadly, the result suggests that the picture of kernelizability of $H$-free Edge Modification problems could be more complex than conjectured by Marx and Sandeep. If so, the question of precisely where the tractability line goes for polynomial kernelization seems highly challenging, as all kernelization results so far (including ours) rely on highly case-specific structural characterizations of $H$-free graphs. 
We leave these deeper investigations into the problem open for future work. 
We also leave open the question of a polynomial kernel for \textsc{Prison-free Edge Editing}. 


\bibliography{refs}



\end{document}